\documentclass[letterpaper,twocolumn,twoside]{IEEEtran} 
\usepackage[noadjust]{cite}
\usepackage{amsmath, amssymb, bm, epsfig, psfrag,amsthm}
\usepackage{algorithm,algorithmic}
\usepackage{bbm}
\usepackage[usenames]{color}
\usepackage{tikz}
\usepackage{etoolbox}
\usepackage{enumerate}
\usetikzlibrary{shapes,arrows}
\usetikzlibrary{patterns}
\usepackage{pgfplots}
\usepackage{url,hyperref}
\usepackage{xspace}
\usepackage{enumerate}
\usepackage{mathtools}
\usepackage{epstopdf}

\newcommand{\figsize}{3.15in}

\newtoggle{conference}
\toggletrue{conference} 

\usepackage[normalem]{ulem} 

\newcommand{\textb}[1]{\textcolor{black}{#1}}

\renewcommand{\tilde}{\widetilde}
\renewcommand{\hat}{\widehat}

\def\beq{\begin{equation}}
\def\eeq{\end{equation}}
\def\beqa{\begin{eqnarray}}
\def\eeqa{\end{eqnarray}}
\def\beqan{\begin{eqnarray*}}
\def\eeqan{\end{eqnarray*}}

\def\R{{\mathbb{R}}}

\def\argmin{\mathop{\mathrm{arg\,min}}}
\def\argmax{\mathop{\mathrm{arg\,max}}}
\def\diag{\mathop{\mathrm{diag}}}
\def\x{\times}

\newtheorem{theorem}{Theorem}
\newtheorem{lemma}{Lemma}

\def\xhat{\hat{x}}

\def\MSE{\mbox{\small \sffamily MSE}}

\def\arr{\rightarrow}
\def\Exp{\mathbb{E}}
\def\var{\mathrm{var}}
\def\Cov{\mathrm{Cov}}
\def\Tr{\mathrm{Tr}}
\def\rank{\mathrm{rank}}
\def\alphabar{\overline{\alpha}}
\def\etabar{\overline{\eta}}
\def\gammabar{\overline{\gamma}}
\def\gammatil{\tilde{\gamma}}

\def\tm1{t\! - \! 1}
\def\tp1{t\! + \! 1}
\def\km1{k\! - \! 1}
\def\kp1{k\! + \! 1}
\def\ip1{i\! + \! 1}
\def\im1{i\! - \! 1}

\newcommand{\zero}{\mathbf{0}}
\newcommand{\one}{\mathbf{1}}
\newcommand{\abf}{\mathbf{a}}
\newcommand{\bbf}{\mathbf{b}}
\newcommand{\dbf}{\mathbf{d}}

\newcommand{\fbf}{\mathbf{f}}
\newcommand{\gbf}{\mathbf{g}}

\newcommand{\pbf}{\mathbf{p}}

\newcommand{\qbf}{\mathbf{q}}

\newcommand{\rbf}{\mathbf{r}}

\newcommand{\rbftil}{\tilde{\mathbf{r}}}
\newcommand{\sbf}{\mathbf{s}}
\newcommand{\sbfbar}{\overline{\mathbf{s}}}

\newcommand{\ubf}{\mathbf{u}}

\newcommand{\vbf}{\mathbf{v}}

\newcommand{\wbf}{\mathbf{w}}

\newcommand{\xbf}{\mathbf{x}}
\newcommand{\xbfhat}{\hat{\mathbf{x}}}
\newcommand{\xbftil}{\tilde{\mathbf{x}}}
\newcommand{\ybf}{\mathbf{y}}

\newcommand{\ybftil}{\tilde{\mathbf{y}}}

\newcommand{\Abf}{\mathbf{A}}
\newcommand{\Bbf}{\mathbf{B}}

\newcommand{\Cbf}{\mathbf{C}}

\newcommand{\Ibf}{\mathbf{I}}

\newcommand{\Pbf}{\mathbf{P}}

\newcommand{\Qbf}{\mathbf{Q}}

\newcommand{\Sbf}{\mathbf{S}}
\newcommand{\Ubf}{\mathbf{U}}

\newcommand{\Ubfbar}{\overline{\mathbf{U}}}
\newcommand{\Vbf}{\mathbf{V}}

\newcommand{\Vbfbar}{\overline{\mathbf{V}}}
\newcommand{\Wbf}{\mathbf{W}}
\newcommand{\Xhat}{\hat{X}}

\def\xibf{{\boldsymbol \xi}}

\newcommand{\phibf}{{\bm{\phi}}}

\newcommand{\indic}[1]{\mathbbm{1}_{ \{ {#1} \} }}
\newcommand{\mat}[1]{\ensuremath{\begin{bmatrix}#1\end{bmatrix}}}

\newcommand{\MAP}{_{\text{\sf MAP}}}
\newcommand{\MMSE}{_{\text{\sf MMSE}}}

\newcommand{\tran}{^{\text{\sf T}}}

\def\eqd{\stackrel{d}{=}}
\def\PLeq{\stackrel{PL(2)}{=}}
\def\Norm{{\mathcal N}}
\def\Range{\mathrm{Range}}
\def\Diag{\mathrm{Diag}}
\def\alphabar{\overline{\alpha}}

\def\Ecal{{\mathcal E}}
\def\Kit{K_{\rm it}}

\newcommand*\dif{\mathop{}\!\mathrm{d}} 

\newcommand{\bkt}[1]{{\langle #1 \rangle}}
\newcommand{\msg}[2]{\mu_{#1\rightarrow#2}}

\newcommand{\belapp}{b_{\textsf{app}}}
\newcommand{\beltil}{b_{\textsf{sp}}}

\tikzstyle{block}=[rectangle,draw, fill=blue!20,
    minimum height=1cm, minimum width=1.5cm]
\tikzstyle{signal}=[coordinate,draw]

\title{Vector Approximate Message Passing}

\author{
     Sundeep Rangan, \IEEEmembership{Fellow,~IEEE},
     Philip Schniter, \IEEEmembership{Fellow,~IEEE},
     and
     Alyson K.~Fletcher, \IEEEmembership{Member,~IEEE}.
     \thanks{S.~Rangan (email: srangan@nyu.edu) is with
           the Department of Electrical and Computer Engineering,
           New York University, Brooklyn, NY, 11201.
           His work was supported by the National Science
           Foundation under Grants 1302336, 1564142, and 1547332.}
     \thanks{P.~Schniter (email: schniter.1@osu.edu) is with
           the Department of Electrical and Computer Engineering,
           The Ohio State University,
           Columbus, OH, 43210.
           His work was supported in part by
           the National Science Foundation grant
           CCF-1527162.}
     \thanks{A.~K.~Fletcher (email: akfletcher@ucla.edu) is with
           the Departments of Statistics, Mathematics, and Electrical
           Engineering, University of California, Los Angeles, CA, 90095.
           Her work is supported in part by National Science Foundation
           grants 1254204 and 1564278 as well as
           the Office of Naval Research grant N00014-15-1-2677.}
}

\begin{document}
\setlength{\arraycolsep}{0.8mm}

\maketitle
\begin{abstract}
The standard linear regression (SLR) problem is to
recover a vector $\mathbf{x}^0$ from
noisy linear observations $\mathbf{y}=\mathbf{Ax}^0+\mathbf{w}$.
The approximate message passing (AMP) algorithm proposed
by Donoho, Maleki, and Montanari is a computationally efficient
iterative approach to SLR that has a remarkable property:
for large i.i.d.\ sub-Gaussian matrices $\mathbf{A}$, its per-iteration behavior
is rigorously characterized by a scalar state-evolution
whose fixed points, when unique, are Bayes optimal.
\textb{The AMP algorithm}, however, is fragile in that even small deviations from
the i.i.d.\ sub-Gaussian model can cause the algorithm to diverge.
This paper considers a ``vector AMP'' (VAMP) algorithm and shows that
VAMP has a rigorous
scalar state-evolution that holds under a much broader class of large
random matrices $\mathbf{A}$: those that are right-\textb{orthogonally} invariant.
After performing an initial singular value decomposition (SVD) of $\mathbf{A}$,
the per-iteration complexity of VAMP is similar to that of AMP.
In addition, the fixed points of VAMP's state evolution are
consistent with the replica prediction of the minimum mean-squared
error derived by Tulino, Caire, Verd\'u, and Shamai.
\textb{Numerical experiments are used to confirm the effectiveness of VAMP and its consistency with state-evolution predictions.}
\end{abstract}


\section{Introduction} \label{sec:intro}

Consider the problem of recovering a vector $\xbf^0 \in \R^N$
from noisy linear measurements of the form
\beq \label{eq:yAx}
    \ybf = \Abf\xbf^0 + \wbf \in \R^M ,
\eeq
where $\Abf$ is a known matrix and $\wbf$ is an unknown, unstructured noise vector.
In the statistics literature, this problem is known as \emph{standard linear regression}, and in the signal processing literature this is known as \emph{solving a linear inverse problem}, or as \emph{compressive sensing} when $M\ll N$ and $\xbf^0$ is sparse.

\subsection{Problem Formulations} \label{sec:formulations}

One approach to recovering $\xbf^0$ is \emph{regularized quadratic loss minimization},
where an estimate $\xbfhat$ of $\xbf^0$ is computed by solving an optimization problem of the form
\beq
\xbfhat
= \argmin_{\xbf\in\R^N} \frac{1}{2}\|\ybf-\Abf\xbf\|_2^2 + f(\xbf) .
\label{eq:RQLM}
\eeq
Here, the penalty function or ``regularization'' $f(\xbf)$ is chosen to promote a desired structure in $\xbfhat$.
For example, the choice $f(\xbf)=\lambda\|\xbf\|_1$ with $\lambda>0$ promotes sparsity in $\xbfhat$.

Another approach is through the Bayesian methodology.
Here, one presumes a prior density $p(\xbf)$ and likelihood function $p(\ybf|\xbf)$ and then aims to compute the posterior density
\beq \label{eq:Bayes}
p(\xbf|\ybf) = \frac{p(\ybf|\xbf)p(\xbf)}{\int p(\ybf|\xbf)p(\xbf)\dif\xbf}
\eeq
or, in practice, a summary of it \cite{pereyra2016stoch}.
Example summaries include the maximum \emph{a posteriori} (MAP) estimate
\beq \label{eq:MAP}
\xbfhat\MAP = \argmax_\xbf p(\xbf|\ybf) ,
\eeq
the minimum mean-squared error (MMSE) estimate
\beq \label{eq:MMSE}
\xbfhat\MMSE = \argmin_{\xbftil} \int \|\xbf-\xbftil\|^2 p(\xbf|\ybf) \dif\xbf
= \Exp[\xbf|\ybf] ,
\eeq
or the posterior marginal densities $\{p(x_n|\ybf)\}_{n=1}^N$.

Note that, if the noise $\wbf$ is modeled as
$\wbf \sim \Norm(\zero,\gamma_w^{-1}\Ibf)$,
i.e., additive white Gaussian noise (AWGN) with some precision $\gamma_w>0$,
then
the regularized quadratic loss minimization problem \eqref{eq:RQLM}
is equivalent to
MAP estimation under the prior
$p(\xbf) \propto \exp[ -\gamma_w f(\xbf) ]$,
where $\propto$ denotes equality up to a scaling that is independent of $\xbf$.
Thus we focus on MAP, MMSE, and marginal posterior inference in the sequel.

\subsection{Approximate Message Passing} \label{sec:amp}

Recently, the so-called \emph{approximate message passing} (AMP) algorithm \cite{DonohoMM:09,DonohoMM:10-ITW1} was proposed as an iterative method to recover $\xbf^0$ from measurements of the form \eqref{eq:yAx}.
The AMP iterations are specified in Algorithm~\ref{algo:amp}.
There,\footnote{The subscript ``1'' in $\gbf_1$ is used promote notational consistency with Vector AMP algorithm presented in the sequel.}
$\gbf_1(\cdot,\gamma_k):\R^N\rightarrow\R^N$ is a \emph{denoising} function parameterized by $\gamma_k$,
and $\bkt{\gbf_1'(\rbf_k,\gamma_k)}$ is its \emph{divergence} at $\rbf_k$.
In particular, $\gbf_1'(\rbf_k,\gamma_k)\in\R^N$ is the diagonal of the Jacobian,
\begin{align}
\gbf_1'(\rbf_k,\gamma_k)
= \diag\left[\frac{\partial \gbf_1(\rbf_k,\gamma_k)}{\partial \rbf_k}\right]
\label{eq:jacobian} ,
\end{align}
and $\bkt{\cdot}$ is the empirical averaging operation
\begin{align}
\bkt{\ubf} := \frac{1}{N} \sum_{n=1}^N u_n
\label{eq:bkt} .
\end{align}

\begin{algorithm}[t]
\caption{AMP}
\begin{algorithmic}[1]  \label{algo:amp}
\REQUIRE{Matrix $\Abf\!\in\!\R^{M\times N}$, measurement vector $\ybf$,
denoiser $\gbf_1(\cdot,\gamma_k)$, and number of iterations $\Kit$.}
\STATE{Set $\vbf_{-1}=\zero$ and select initial $\rbf_0,\gamma_0$.}
\FOR{$k=0,1,\dots,\Kit$}
    \STATE{$\xbfhat_{k} = \gbf_1(\rbf_k,\gamma_k)$}
        \label{line:x}
    \STATE{$\alpha_k = \bkt{\gbf_1'(\rbf_k,\gamma_k)}$}
        \label{line:a}
    \STATE{$\vbf_k = \ybf - \Abf\xbfhat_k
        + \frac{N}{M}\alpha_{k-1}\vbf_{k-1}$}
        \label{line:v}
    \STATE{$\rbf_{\kp1} = \xbfhat_k + \Abf\tran\vbf_k$}
        \label{line:r}
    \STATE{Select $\gamma_{\kp1}$}
        \label{line:gamma}
\ENDFOR
\STATE{Return $\xbfhat_{\Kit}$.}
\end{algorithmic}
\end{algorithm}

When $\Abf$ is a large i.i.d.\ sub-Gaussian matrix,
$\wbf\sim\Norm(\zero,\gamma_{w0}^{-1}\Ibf)$,
and $\gbf_1(\cdot,\gamma_k)$ is \emph{separable}, i.e.,
\begin{align}
[\gbf_1(\rbf_k,\gamma_k)]_n
= g_1(r_{kn},\gamma_k)
~ \forall n
\label{eq:g1sep} ,
\end{align}
with identical Lipschitz components $g_1(\cdot,\gamma_k):\R\rightarrow \R$,
AMP displays a remarkable behavior, which is that $\rbf_k$ behaves like a white-Gaussian-noise corrupted version of the true signal $\xbf^0$ \cite{DonohoMM:09}.
That is,
\begin{align}
\rbf_k
&= \xbf^0 + \Norm(\zero,\tau_k\Ibf) ,
\label{eq:unbiased}
\end{align}
for some variance $\tau_k>0$.
Moreover, the variance $\tau_k$ can be predicted through the following \emph{state evolution} (SE):
\begin{subequations} \label{eq:ampSE}
\begin{align}
\Ecal(\gamma_k,\tau_k)
&= \frac{1}{N}\Exp\left[\big\|\gbf_1\big(\xbf^0+\Norm(\zero,\tau_k\Ibf),\gamma_k\big)-\xbf^0\big\|^2\right] \\
\tau_{\kp1}
&= \gamma_{w0}^{-1} + \frac{N}{M}\Ecal(\gamma_k,\tau_k),
\end{align}
\end{subequations}
where $\Ecal(\gamma_k,\tau_k)$ is the \textb{MSE of the} AMP estimate $\xbfhat_k$.

The AMP SE \eqref{eq:ampSE} was rigorously established
for i.i.d.\ Gaussian $\Abf$ in \cite{BayatiM:11} and for i.i.d.\ sub-Gaussian $\Abf$ in \cite{BayLelMon:15}
in the \emph{large-system limit} (i.e., $N,M\rightarrow\infty$ and $N/M\rightarrow\delta\in (0,1)$) under some mild regularity conditions.
Because the SE \eqref{eq:ampSE} holds for generic $g_1(\cdot,\gamma_k)$ and generic $\gamma_k$-update rules, it can be used to characterize the application of AMP to many problems, as further discussed in Section~\ref{sec:bayes}.

\subsection{Limitations, Modifications, and Alternatives to AMP} \label{sec:alternatives}

An important limitation of AMP's SE is that it holds only under large i.i.d.\ sub-Gaussian $\Abf$.
Although recent analysis \cite{rush2016finite} has rigorously analyzed AMP's performance under finite-sized i.i.d.\ Gaussian $\Abf$, there remains the important question of how AMP behaves with general $\Abf$.

Unfortunately, it turns out that the AMP Algorithm~\ref{algo:amp} is somewhat fragile with regard to the construction of $\Abf$.
For example, AMP diverges with even mildly ill-conditioned or non-zero-mean $\Abf$ \cite{RanSchFle:14-ISIT,Caltagirone:14-ISIT,Vila:ICASSP:15}.
Although damping \cite{RanSchFle:14-ISIT,Vila:ICASSP:15}, mean-removal \cite{Vila:ICASSP:15}, sequential updating \cite{manoel2015swamp}, and direct free-energy minimization \cite{rangan2015admm} all help to prevent AMP from diverging, such strategies are limited in effectiveness.

Many other algorithms for standard linear regression \eqref{eq:yAx} have been designed using approximations of belief propagation (BP) and/or free-energy minimization.
Among these are the Adaptive Thouless-Anderson-Palmer (ADATAP) \cite{opper2001adaptive}, Expectation Propagation (EP) \cite{Minka:01,seeger2005expectation}, Expectation Consistent Approximation (EC) \cite{OppWin:05,kabashima2014signal,fletcher2016expectation}, (S-transform AMP) S-AMP \cite{cakmak2014samp,cakmak2015samp},
and (Orthogonal AMP) OAMP \cite{ma2016orthogonal} approaches.
Although numerical experiments suggest that some of these algorithms are more robust than AMP Algorithm~\ref{algo:amp} to the choice of $\Abf$, their convergence has not been rigorously analyzed.
In particular, there remains the question of whether there exists an AMP-like algorithm with a rigorous SE analysis that holds for a larger class of matrices than i.i.d.\ sub-Gaussian.
In the sequel, we describe one such algorithm.

\subsection{Contributions}

In this paper, we propose a computationally efficient iterative algorithm for the estimation of the vector $\xbf^0$ from noisy linear measurements $\ybf$ of the form in \eqref{eq:yAx}.
(See Algorithm~\ref{algo:vampSVD}.)
We call the algorithm ``\emph{vector AMP}'' (VAMP) because
i) its behavior can be rigorously characterized by a scalar SE under large random $\Abf$,
and
ii) it can be derived using an approximation of BP on a factor graph with vector-valued variable nodes.
We outline VAMP's derivation in Section~\ref{sec:vamp} with the aid of some background material that is reviewed in Section~\ref{sec:back}.

In Section~\ref{sec:SE}, we establish the VAMP SE in the case of
large \emph{right-orthogonally invariant} random $\Abf$
and separable Lipschitz denoisers $\gbf_1(\cdot,\gamma_k)$,
using techniques similar to those used by Bayati and Montanari in \cite{BayatiM:11}.
Importantly, these right-orthogonally invariant $\Abf$ allow arbitrary singular values and arbitrary left singular vectors, making VAMP much more robust than AMP in regards to the construction of $\Abf$.
In Section~\ref{sec:replica}, we establish that the asymptotic MSE predicted by VAMP's SE agrees with the MMSE predicted by the replica method \cite{tulino2013support} when VAMP's priors are matched to the true data.
Finally, in Section~\ref{sec:num}, we present numerical experiments demonstrating that VAMP's empirical behavior matches its SE at moderate dimensions, even when $\Abf$ is highly ill-conditioned or non-zero-mean.

\subsection{Relation to Existing Work}

The idea to construct algorithms from graphical models with vector-valued nodes is not new, and in fact underlies the EC- and EP-based algorithms described in~\cite{Minka:01,seeger2005expectation,OppWin:05,kabashima2014signal,fletcher2016expectation}.
The use of vector-valued nodes is also central to the derivation of S-AMP~\cite{cakmak2014samp,cakmak2015samp}.
In the sequel, we present a simple derivation of VAMP \textb{that uses the EP methodology from \cite{Minka:01,seeger2005expectation}, which passes approximate messages between the nodes of a factor graph.
But we note that VAMP can also be derived using the EC methodology, which formulates a variational optimization problem using a constrained version of the Kullback-Leibler distance and then relaxes the density constraints to moment constraints.
For more details on the latter approach, we refer the interested reader to the discussion of ``diagonal restricted EC'' in \cite[App.~D]{OppWin:05} and ``uniform diagonalized EC'' in \cite{fletcher2016expectation}.
}

It was recently shown \cite{kabashima2014signal} that, for large right-orthogonally invariant $\Abf$, the fixed points of diagonal-restricted EC are ``good'' in the sense that they are consistent with a certain replica prediction of the MMSE that is derived in \cite{kabashima2014signal}.
Since the fixed points of ADATAP and S-AMP are known \cite{cakmak2014samp} to coincide with those of diagonal-restricted EC (and thus VAMP), all of these algorithms can be understood to have good fixed points.
The trouble is that these algorithms do not necessarily converge to their fixed points.
For example, S-AMP diverges with even mildly ill-conditioned or non-zero-mean $\Abf$, as demonstrated in Section~\ref{sec:num}.
\emph{Our main contribution is establishing that VAMP's behavior can be exactly predicted by an SE analysis analogous to that for AMP.
This SE analysis then provides precise convergence guarantees for large right-orthogonally invariant $\Abf$.}
The numerical results presented in Section~\ref{sec:num} confirm that, in practice, VAMP's convergence is remarkably robust, even with very ill-conditioned or mean-perturbed matrices $\Abf$ of finite dimension.

The main insight that leads to both the VAMP algorithm and its SE analysis
comes from a consideration of the singular value decomposition (SVD) of $\Abf$.
Specifically, take the ``economy" SVD,
\begin{align}
\Abf
&=\Ubfbar\Diag(\sbfbar)\Vbfbar\tran
\label{eq:econSVD} ,
\end{align}
where $\sbfbar\in\R^R$ for $R:=\rank(\Abf)\leq\min(M,N)$.
The VAMP iterations can be performed by matrix-vector multiplications with $\Vbfbar\in\R^{N\times R}$ and $\Vbfbar\tran$, yielding a structure
very similar to that of AMP.
Computationally, the SVD form of VAMP
(i.e., Algorithm~\ref{algo:vampSVD})
has the benefit that, once the SVD has been computed,
VAMP's per-iteration cost will be dominated by $O(RN)$ floating-point operations (flops), as opposed to $O(N^3)$ for the EC methods from
\cite[App.~D]{OppWin:05} or \cite{fletcher2016expectation}.
Furthermore, if these matrix-vector multiplications have fast implementations (e.g., $O(N)$ when $\Vbfbar$ is a discrete wavelet transform), then the complexity of VAMP reduces accordingly.
We emphasize that VAMP uses a single SVD, not a per-iteration SVD.
In many applications, this SVD can be computed off-line.
In the case that SVD complexity may be an issue, we note that it costs $O(MNR)$ flops
by classical methods or $O(MN\log R)$ by modern approaches \cite{halko2011matrix}.

The SVD offers more than just a fast algorithmic implementation.
More importantly, it connects VAMP to AMP in such a way that the Bayati and Montanari's SE analysis of AMP \cite{BayatiM:11} can be extended to obtain a rigorous SE for VAMP.
In this way, the SVD can be viewed as a proof technique.
Since it will be useful for derivation/interpretation in the sequel, we note that the VAMP iterations can also be written without an explicit SVD
(see Algorithm~\ref{algo:vamp}),
in which case they coincide with the uniform-diagonalization variant of the generalized EC method from \cite{fletcher2016expectation}.
In this latter implementation, the linear MMSE (LMMSE) estimate \eqref{eq:g2slr} must be computed at each iteration, as well as the trace of its covariance matrix \eqref{eq:a2slr}, which both involve the inverse of an $N\times N$ matrix.

The OAMP-LMMSE algorithm from \cite{ma2016orthogonal}
is similar to VAMP and diagonal-restricted EC,
but different in that it approximates certain variance terms.
This difference can be seen by comparing
equations (30)-(31) in \cite{ma2016orthogonal} to
lines~\ref{line:gamtilsvd} and \ref{line:gamsvd} in Algorithm~\ref{algo:vampSVD}
(or lines~\ref{line:gam1} and \ref{line:gam2} in Algorithm~\ref{algo:vamp}).
Furthermore, OAMP-LMMSE differs from VAMP in its reliance on matrix inversion (see, e.g., the comments in the Conclusion of \cite{ma2016orthogonal}).


\textb{Shortly after the initial publication of this work, \cite{takeuchi2017rigorous} proved a very similar result for the complex case using a fully probabilistic analysis.}

\subsection{Notation}

We use
capital boldface letters like $\Abf$ for matrices,
small boldface letters like $\abf$ for vectors,
$(\cdot)\tran$ for transposition,
and $a_n=[\abf]_n$ to denote the $n$th element of $\abf$.
Also, we use
$\|\abf\|_p=(\sum_n |a_n|^p)^{1/p}$ for the $\ell_p$ norm of $\abf$,
$\|\Abf\|_2$ for the spectral norm of $\Abf$,
$\Diag(\abf)$ for the diagonal matrix created from vector $\abf$, and
$\diag(\Abf)$ for the vector extracted from the diagonal of matrix $\Abf$.
Likewise, we use
$\Ibf_N$ for the $N\times N$ identity matrix,
$\zero$ for the matrix of all zeros, and
$\one$ for the matrix of all ones.
For a random vector $\xbf$, we denote
its probability density function (pdf) by $p(\xbf)$,
its expectation by $\Exp[\xbf]$,
and
its covariance matrix by $\Cov[\xbf]$.
Similarly, we use
$p(\xbf|\ybf)$, $\Exp[\xbf|\ybf]$, and $\Cov[\xbf|\ybf]$ for the
\emph{conditional} pdf, expectation, and covariance, respectively.
Also, we use
$\Exp[\xbf|b]$ and $\Cov[\xbf|b]$ to denote the
expectation and covariance of $\xbf\sim b(\xbf)$,
i.e., $\xbf$ distributed according to the pdf $b(\xbf)$.
We refer to
the Dirac delta pdf using $\delta(\xbf)$
and to
the pdf of a Gaussian random vector $\xbf\in\R^N$ with mean $\abf$ and covariance $\Cbf$ using $\Norm(\xbf;\abf,\Cbf)=\exp( -(\xbf-\abf)\tran\Cbf^{-1}(\xbf-\abf)/2 )/\sqrt{(2\pi)^N|\Cbf|}$.
Finally,
$p(\xbf)\propto f(\xbf)$ says that functions $p(\cdot)$ and $f(\cdot)$ are equal up to a scaling that is invariant to $\xbf$.

\section{Background on the AMP Algorithm} \label{sec:back}

In this section, we provide background on the AMP algorithm that will be useful in the sequel.

\subsection{Applications to Bayesian Inference} \label{sec:bayes}

We first detail the application of the AMP Algorithm~\ref{algo:amp} to the Bayesian inference problems from Section~\ref{sec:formulations}.
Suppose that the prior on $\xbf$ is i.i.d., so that it takes the form
\beq
p(\xbf)=\prod_{n=1}^N p(x_n)
\label{eq:pxiid} .
\eeq
Then AMP can be applied to MAP problem \eqref{eq:MAP} by choosing the scalar denoiser as
\beq
g_1(r_{kn},\gamma_k)
= \argmin_{x_n\in\R} \left[ \frac{\gamma_k}{2}|x_n-r_{kn}|^2 - \ln p(x_n) \right]
\label{eq:gmapsca} .
\eeq
Likewise, AMP can be applied to the MMSE problem \eqref{eq:MMSE} by choosing
\beq
g_1(r_{kn},\gamma_k)
= \Exp[x_n|r_{kn},\gamma_k]
\label{eq:g1mmsesca} ,
\eeq
where the expectation in \eqref{eq:g1mmsesca} is with respect to the conditional density
\begin{align}
p(x_n|r_{kn},\gamma_k)
\propto \exp\left[ -\frac{\gamma_k}{2}|r_{kn}-x_n|^2 +\ln p(x_n) \right]
\label{eq:pxr1sca} .
\end{align}
In addition, $p(x_n|r_{kn},\gamma_k)$ in \eqref{eq:pxr1sca} acts as AMP's iteration-$k$ approximation of the marginal posterior $p(x_n|\ybf)$.
For later use, we note that the derivative of the MMSE scalar denoiser \eqref{eq:g1mmsesca} w.r.t.\ its first argument can be expressed as
\beq
    g_1'(r_{kn},\gamma_k) = \gamma_k\var\left[ x_n | r_{kn},\gamma_k \right]
\label{eq:g1dervar} ,
\eeq
where the variance is computed with respect to the density~\eqref{eq:pxr1sca}
(see, e.g., \cite{Rangan:11-ISIT}).

In \eqref{eq:gmapsca}-\eqref{eq:pxr1sca}, $\gamma_k$ can be interpreted as an estimate of $\tau_k^{-1}$, the iteration-$k$ precision of $\rbf_k$ from \eqref{eq:unbiased}.
In the case that $\tau_k$ is known, the ``matched'' assignment
\begin{align}
\gamma_k=\tau_k^{-1}
\label{eq:gamma_matched}
\end{align}
leads to the interpretation of \eqref{eq:gmapsca} and \eqref{eq:g1mmsesca} as the scalar MAP and MMSE denoisers of $r_{kn}$, respectively.
Since, in practice, $\tau_k$ is usually not known, it has been suggested to use
\begin{align}
\gamma_{\kp1} &= \frac{M}{\|\vbf_k\|^2}
\label{eq:gamma_amp_practical} ,
\end{align}
although other choices are possible \cite{Montanari:12-bookChap}.

\subsection{Relation of AMP to IST} \label{sec:ista}

The AMP Algorithm~\ref{algo:amp} is closely related to the well-known \emph{iterative soft thresholding} (IST) algorithm \cite{ChamDLL:98,DaubechiesDM:04} that can be used\footnote{The IST algorithm is guaranteed to converge \cite{DaubechiesDM:04} when $\|\Abf\|_2<1$.} to solve \eqref{eq:RQLM} with convex $f(\cdot)$.
In particular, if the term
\begin{align}
\frac{N}{M}\alpha_{k-1}\vbf_{k-1}
\label{eq:onsager}
\end{align}
is removed from line~\ref{line:v} of Algorithm~\ref{algo:amp}, then what remains is the IST algorithm.

The term \eqref{eq:onsager} is known as the \emph{Onsager} term in the statistical physics literature \cite{ThoulessAP:77}.
Under large i.i.d.\ sub-Gaussian $\Abf$, the Onsager correction ensures the behavior in \eqref{eq:unbiased}.
When \eqref{eq:unbiased} holds, the denoiser $g_{1}(\cdot,\gamma_k)$ can be optimized accordingly, in which case each iteration of AMP becomes very productive.
As a result, AMP converges much faster than ISTA for i.i.d.\ Gaussian $\Abf$ (see, e.g., \cite{Montanari:12-bookChap} for a comparison).

\subsection{Derivations of AMP} \label{sec:amp_deriv}

The AMP algorithm can be derived in several ways.
One way is through approximations of loopy belief propagation (BP) \cite{Pearl:88,YedidiaFW:03} on a bipartite factor graph constructed from the factorization
\begin{align}
p(\ybf,\xbf)
&= \left[ \prod_{m=1}^M \Norm(y_m;\abf_m\tran\xbf,\gamma_w^{-1}) \right] \left[ \prod_{n=1}^N p(x_n) \right]
\label{eq:amp_factors} ,
\end{align}
where $\abf_m\tran$ denotes the $m$th row of $\Abf$.
We refer the reader to \cite{DonohoMM:10-ITW1,Rangan:11-ISIT} for details on the message-passing derivation of AMP, noting connections to the general framework of \emph{expectation propagation} (EP) \cite{Minka:01,seeger2005expectation}.
AMP can also be derived through a ``free-energy'' approach, where one
i) proposes a cost function, ii) derives conditions on its stationary points, and iii) constructs an algorithm whose fixed points coincide with those stationary points.
We refer the reader to \cite{RanSRFC:13-ISIT,Krzakala:14-ISITbethe,cakmak2014samp} for details, and note connections to the general framework of \emph{expectation consistent approximation} (EC) \cite{OppWin:05,fletcher2016expectation}.

\section{The Vector AMP Algorithm} \label{sec:vamp}

The \emph{Vector AMP} (VAMP) algorithm is stated in Algorithm~\ref{algo:vampSVD}.
In line~\ref{line:dsvd}, ``$\sbfbar^2$'' refers to the componentwise square of vector $\sbfbar$.
Also, $\Diag(\abf)$ denotes the diagonal matrix whose diagonal components are given by the vector $\abf$.

\begin{algorithm}[t]
\caption{Vector AMP (SVD Form)}
\begin{algorithmic}[1]  \label{algo:vampSVD}
\REQUIRE{
Matrix $\Abf \in \R^{M \x N}$; measurements $\ybf \in \R^M$;
denoiser $\gbf_1(\cdot,\gamma_k)$;
assumed noise precision $\gamma_w \geq 0$; and
number of iterations $\Kit$. }
\STATE{Compute economy SVD $\Ubfbar\Diag(\sbfbar)\Vbfbar\tran=\Abf$
with $\Ubfbar\tran\Ubfbar=\Ibf_R$,
$\Vbfbar\tran\Vbfbar=\Ibf_R$,
$\sbfbar\in\R_+^{R}$},
$R=\rank(\Abf)$.
\STATE{Compute preconditioned $\ybftil:=\Diag(\sbfbar)^{-1}\Ubfbar\tran\ybf$}
\STATE{Select initial $\rbf_{0}$ and $\gamma_{0}\geq 0$.}
\FOR{$k=0,1,\dots,\Kit$}
    \STATE{$\xbfhat_{k} = \gbf_1(\rbf_{k},\gamma_{k})$}
        \label{line:xsvd}
    \STATE{$\alpha_{k} = \bkt{ \gbf_1'(\rbf_{k},\gamma_{k}) }$}
        \label{line:asvd}
    \STATE{$\rbftil_k = (\xbfhat_{k} - \alpha_{k}\rbf_{k})/(1-\alpha_{k})$}
        \label{line:xtilsvd}
    \STATE{$\gammatil_{k} = \gamma_{k}(1-\alpha_{k})/\alpha_{k}$}
        \label{line:gamtilsvd}
    \STATE{$\dbf_k=\gamma_w\Diag\big(\gamma_w\sbfbar^2+\gammatil_{k}\one\big)^{-1}\sbfbar^2$}
        \label{line:dsvd}
    \STATE{$\gamma_{\kp1} = \gammatil_{k}\bkt{\dbf_k}/(\frac{N}{R}-\bkt{\dbf_k})$}
        \label{line:gamsvd}
    \STATE{$\rbf_{\kp1} = \rbftil_k + \frac{N}{R}
        \Vbfbar\Diag\big(\dbf_k/\bkt{\dbf_k}\big)\big(\ybftil-\Vbfbar\tran\rbftil_k\big)$}
        \label{line:rsvd}
\ENDFOR
\STATE{Return $\xbfhat_{\Kit}$.}
\end{algorithmic}
\end{algorithm}

\subsection{Relation of VAMP to AMP} \label{sec:vamp2amp}

A visual examination of VAMP Algorithm~\ref{algo:vampSVD} shows many similarities with AMP Algorithm~\ref{algo:amp}.
In particular, the denoising and divergence steps in lines~\ref{line:xsvd}-\ref{line:asvd} of Algorithm~\ref{algo:vampSVD} are identical to those in lines~\ref{line:x}-\ref{line:a} of Algorithm~\ref{algo:amp}.
Likewise, an Onsager term $\alpha_k\rbf_k$ is visible in line~\ref{line:xtilsvd} of Algorithm~\ref{algo:vampSVD}, analogous to the one in line~\ref{line:v} of Algorithm~\ref{algo:amp}.
Finally, the per-iteration computational complexity of each algorithm is dominated by two matrix-vector multiplications: those involving $\Abf$ and $\Abf\tran$ in Algorithm~\ref{algo:amp} and those involving $\Vbfbar$ and $\Vbfbar\tran$ in Algorithm~\ref{algo:vampSVD}.

The most important similarity between the AMP and VAMP algorithms is not obvious from visual inspection and will be established rigorously in the sequel.
It is the following: for certain large random $\Abf$,
the VAMP quantity $\rbf_k$ behaves
like a white-Gaussian-noise corrupted version of the true signal $\xbf^0$, i.e.,
\begin{align}
\rbf_k
&= \xbf^0 + \Norm(\zero,\tau_k\Ibf) ,
\label{eq:unbiased2}
\end{align}
for some variance $\tau_k>0$.
Moreover, the noise variance $\tau_k$ can be tracked through a scalar SE formalism whose details will be provided in the sequel.
Furthermore, the VAMP quantity $\gamma_k$ can be interpreted as an estimate of $\tau_k^{-1}$ in \eqref{eq:unbiased2}, analogous to the AMP quantity $\gamma_k$ discussed around \eqref{eq:gamma_matched}.

It should be emphasized that the class of matrices $\Abf$ under which the VAMP SE holds is much bigger than the class under which the AMP SE holds.
In particular, VAMP's SE holds for large random matrices $\Abf$ whose right singular%
  \footnote{We use several forms of SVD in this paper.
  Algorithm~\ref{algo:vampSVD} uses the ``economy'' SVD
  $\Abf=\Ubfbar\Diag(\sbfbar)\Vbfbar\tran\in\R^{M\times N}$,
  where $\sbfbar\in\R_+^R$ with $R=\rank(\Abf)$, so that
  $\Ubfbar$ and/or $\Vbfbar$ may be tall.
  The discussion in Section~\ref{sec:vamp2amp} uses the ``standard'' SVD
  $\Abf=\Ubf\Sbf\Vbf\tran$, where $\Sbf\in\R^{M\times N}$ and both
  $\Ubf$ and $\Vbf$ are orthogonal.
  Finally, the state-evolution proof in Section~\ref{sec:SE} uses the
  standard SVD on square $\Abf\in\R^{N\times N}$.}
vector matrix $\Vbf\in\R^{N\times N}$ is uniformly distributed on the group of orthogonal matrices.
Notably, VAMP's SE holds for \emph{arbitrary} (i.e., deterministic) left singular vector matrices $\Ubf$ and singular values, apart from some mild regularity conditions that will be detailed in the sequel.
In contrast, AMP's SE is known to hold \cite{BayatiM:11,BayLelMon:15} only for large i.i.d.\ sub-Gaussian matrices $\Abf$, which implies i) random orthogonal $\Ubf$ and $\Vbf$ and ii) a particular distribution on the singular values of $\Abf$.

\subsection{EP Derivation of VAMP}

\textb{As with AMP (i.e., Algorithm~\ref{algo:amp}), VAMP (i.e., Algorithm~\ref{algo:vampSVD}) can be derived in many ways.}
Here we present a very simple derivation based on an EP-like approximation of the sum-product (SP) belief-propagation algorithm.
Unlike the AMP algorithm, whose message-passing derivation uses a loopy factor graph with \emph{scalar}-valued nodes, the VAMP algorithm uses a non-loopy graph with \emph{vector}-valued nodes, hence the name ``vector AMP.''
We note that VAMP can also be derived using the ``diagonal restricted'' or ``uniform diagonalization'' EC approach \cite{OppWin:05,fletcher2016expectation},
but that derivation is much more complicated.

To derive VAMP, we start with the factorization
\begin{align}
p(\ybf,\xbf)
&= p(\xbf) \Norm(\ybf;\Abf\xbf,\gamma_w^{-1}\Ibf) ,
\end{align}
and split $\xbf$ into two identical variables $\xbf_1=\xbf_2$, giving an equivalent factorization
\begin{align}
p(\ybf,\xbf_1,\xbf_2)
&= p(\xbf_1) \delta(\xbf_1-\xbf_2) \Norm(\ybf;\Abf\xbf_2,\gamma_w^{-1}\Ibf)
\label{eq:vamp_factors} ,
\end{align}
where $\delta(\cdot)$ is the Dirac delta distribution.
The factor graph corresponding to \eqref{eq:vamp_factors} is shown in Figure~\ref{fig:fg_split}.
\begin{figure}[t]
  \centering
  \newcommand{\sz}{0.8}
  \psfrag{px}[b][Bl][\sz]{$p(\xbf_1)$}
  \psfrag{x1}[t][Bl][\sz]{$\xbf_1$}
  \psfrag{del}[b][Bl][\sz]{$\delta(\xbf_1-\xbf_2)$}
  \psfrag{x2}[t][Bl][\sz]{$\xbf_2$}
  \psfrag{py|x}[b][Bl][\sz]{$\Norm(\ybf;\Abf\xbf_2,\gamma_w^{-1}\Ibf)$}
  \includegraphics[width=2.0in]{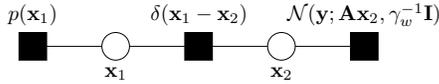}
  \caption{The factor graph used for the derivation of VAMP.
           The circles represent variable nodes and
           the squares represent factor nodes from \eqref{eq:vamp_factors}.}
  \label{fig:fg_split}
\end{figure}
We then pass messages on this factor graph according to the following rules.
\begin{enumerate}
\item \label{rule:b}
\emph{\underline{Approximate beliefs}:}
The approximate belief $\belapp(\xbf)$ on variable node $\xbf$
is $\Norm(\xbf;\xbfhat,\eta^{-1}\Ibf)$, where
$\xbfhat = \Exp[\xbf|\beltil]$ and
$\eta^{-1} = \bkt{\diag(\Cov[\xbf|\beltil])}$
are the mean and average variance of the corresponding SP belief
$\beltil(\xbf) \propto \prod_i \msg{f_i}{\xbf}(\xbf)$,
i.e., the normalized product of all messages impinging on the node.
See Figure~\ref{fig:ep_rules}(a) for an illustration.

\item \label{rule:v2f}
\emph{\underline{Variable-to-factor messages}:}
The message from
a variable node $\xbf$ to a connected factor node $f_i$ is
$\msg{\xbf}{f_i}(\xbf) \propto \belapp(\xbf)/\msg{f_i}{\xbf}(\xbf)$,
i.e., the ratio of the most recent approximate belief $\belapp(\xbf)$ to
the most recent message from $f_i$ to $\xbf$.
See Figure~\ref{fig:ep_rules}(b) for an illustration.

\item \label{rule:f2v}
\emph{\underline{Factor-to-variable messages}:}
The message from a factor node $f$ to a connected variable node
$\xbf_i$ is
$\msg{f}{\xbf_i}(\xbf_i)\propto
 \int f(\xbf_i,\{\xbf_j\}_{j\neq i}\}) \prod_{j\neq i} \msg{\xbf_j}{f}(\xbf_j) \dif\xbf_j$.
See Figure~\ref{fig:ep_rules}(c) for an illustration.
\end{enumerate}
\begin{figure}[t]
  \centering
  \newcommand{\sz}{0.8}
  \newcommand{\szz}{0.6}
  \psfrag{x}[b][Bl][\sz]{$\xbf$}
  \psfrag{x1}[b][Bl][\sz]{$\xbf_2$}
  \psfrag{x2}[b][Bl][\sz]{$\xbf_3$}
  \psfrag{y}[b][Bl][\sz]{$\xbf_1$}
  \psfrag{f}[bl][Bl][0.7]{$f(\xbf_1,\xbf_2,\xbf_3)$}
  \psfrag{f1}[b][Bl][\sz]{$f_1(\xbf)$}
  \psfrag{f2}[b][Bl][\sz]{$f_2(\xbf)$}
  \psfrag{f3}[b][Bl][\sz]{$f_3(\xbf)$}
  \psfrag{m1}[t][Bl][\szz]{$\msg{f_1}{\xbf}(\xbf)$}
  \psfrag{m2}[b][Bl][\szz]{$\msg{f_2}{\xbf}(\xbf)$}
  \psfrag{m3}[t][Bl][\szz]{$\msg{f_3}{\xbf}(\xbf)$}
  \psfrag{m6}[t][Bl][\szz]{$\msg{\xbf}{f_1}(\xbf)$}
  \psfrag{m7}[b][Bl][\szz]{$\msg{\xbf_2}{f}(\xbf_2)$}
  \psfrag{m8}[t][Bl][\szz]{$\msg{\xbf_3}{f}(\xbf_3)$}
  \psfrag{m9}[t][Bl][\szz]{$\msg{f}{\xbf_1}(\xbf_1)$}
  \psfrag{a}[t][Bl][\sz]{(a)}
  \psfrag{b}[t][Bl][\sz]{(b)}
  \psfrag{c}[t][Bl][\sz]{(c)}
  \includegraphics[width=3.2in]{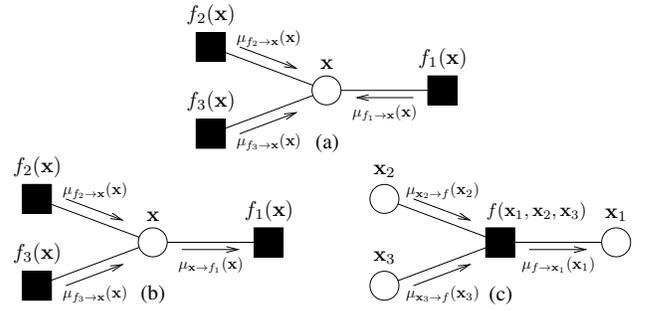}
  \caption{Factor graphs to illustrate
           (a) messaging through a factor node and
           (b) messaging through a variable node.}
  \label{fig:ep_rules}
\end{figure}

By applying the above message-passing rules to the factor graph in Figure~\ref{fig:fg_split}, one obtains Algorithm~\ref{algo:vamp}.
(See Appendix~\ref{sec:EP} for a detailed derivation.)
Lines~\ref{line:x2}--\ref{line:a2} of Algorithm~\ref{algo:vamp} use
\begin{align}
\gbf_2(\rbf_{2k},\gamma_{2k})
&:= \left( \gamma_w \Abf\tran\Abf + \gamma_{2k}\Ibf\right)^{-1}
        \left( \gamma_w\Abf\tran\ybf + \gamma_{2k}\rbf_{2k} \right)
\label{eq:g2slr} ,
\end{align}
which can be recognized as the
MMSE estimate of a random vector $\xbf_2$ under
likelihood $\Norm(\ybf;\Abf\xbf_2,\gamma_w^{-1}\Ibf)$
and prior $\xbf_2\sim \Norm(\rbf_{2k},\gamma_{2k}^{-1}\Ibf)$.
Since this estimate is linear in $\rbf_{2k}$,
we will refer to it as the ``LMMSE'' estimator.
From \eqref{eq:jacobian}-\eqref{eq:bkt} and \eqref{eq:g2slr}, it follows that
line~\ref{line:a2} of Algorithm~\ref{algo:vamp} uses
\begin{align}
\bkt{\gbf'_2(\rbf_{2k},\gamma_{2k})}
&= \frac{\gamma_{2k}}{N} \Tr\left[
        \left( \gamma_w \Abf\tran\Abf + \gamma_{2k}\Ibf\right)^{-1} \right]
\label{eq:a2slr} .
\end{align}

\textb{%
Algorithm~\ref{algo:vamp} is merely a restatement of VAMP Algorithm~\ref{algo:vampSVD}.
Their equivalence can then be seen by substituting the ``economy'' SVD $\Abf=\Ubfbar\Diag(\sbfbar)\Vbfbar\tran$ into Algorithm~\ref{algo:vamp}, simplifying, and equating
$\xbfhat_k \equiv \xbfhat_{1k}$,
$\rbf_k \equiv \rbf_{1k}$,
$\gamma_k \equiv \gamma_{1k}$,
$\gammatil_k \equiv \gamma_{2k}$,
and
$\alpha_k \equiv \alpha_{1k}$.
}

As presented in Algorithm~\ref{algo:vamp}, the steps of VAMP exhibit an elegant symmetry.
The first half of the steps perform denoising on $\rbf_{1k}$ and then Onsager correction in $\rbf_{2k}$, while the second half of the steps perform LMMSE estimation $\rbf_{2k}$ and Onsager correction in $\rbf_{1,\kp1}$.

\begin{algorithm}[t]
\caption{Vector AMP (LMMSE form)}
\begin{algorithmic}[1]  \label{algo:vamp}
\REQUIRE{
LMMSE estimator
$\gbf_2(\rbf_{2k},\gamma_{2k})$ from \eqref{eq:g2slr},
denoiser $\gbf_1(\cdot,\gamma_{1k})$,
and
number of iterations $\Kit$.  }
\STATE{ Select initial $\rbf_{10}$ and $\gamma_{10}\geq 0$.}
\FOR{$k=0,1,\dots,\Kit$}
    \STATE{// Denoising }
    \STATE{$\xbfhat_{1k} = \gbf_1(\rbf_{1k},\gamma_{1k})$}
        \label{line:x1}
    \STATE{$\alpha_{1k} = \bkt{ \gbf_1'(\rbf_{1k},\gamma_{1k}) }$}
        \label{line:a1}
    \STATE{$\eta_{1k} = \gamma_{1k}/\alpha_{1k}$}
        \label{line:eta1}
    \STATE{$\gamma_{2k} = \eta_{1k} - \gamma_{1k}$}
        \label{line:gam2}
    \STATE{$\rbf_{2k} = (\eta_{1k}\xbfhat_{1k} - \gamma_{1k}\rbf_{1k})/\gamma_{2k}$}
        \label{line:r2}
    \STATE{ }
    \STATE{// LMMSE estimation }
    \STATE{$\xbfhat_{2k} = \gbf_2(\rbf_{2k},\gamma_{2k})$}
        \label{line:x2}
    \STATE{$\alpha_{2k} = \bkt{ \gbf_2'(\rbf_{2k},\gamma_{2k}) } $}
        \label{line:a2}
    \STATE{$\eta_{2k} = \gamma_{2k}/\alpha_{2k}$}
        \label{line:eta2}
    \STATE{$\gamma_{1,\kp1} = \eta_{2k} - \gamma_{2k}$}
        \label{line:gam1}
    \STATE{$\rbf_{1,\kp1} = (\eta_{2k}\xbfhat_{2k} - \gamma_{2k}\rbf_{2k})/\gamma_{1,\kp1}$}
        \label{line:r1}
\ENDFOR
\STATE{Return $\xbfhat_{1\Kit}$.}
\end{algorithmic}
\end{algorithm}

\subsection{Implementation Details} \label{sec:implementation}

For practical implementation with finite-dimensional $\Abf$, we find that it helps to make some small enhancements to VAMP.
In the discussion below we will refer to Algorithm~\ref{algo:vampSVD}, but the same approaches apply to Algorithm~\ref{algo:vamp}.

First, we suggest to clip the precisions $\gamma_k$ and $\gammatil_k$ to a positive interval $[\gamma_{\min},\gamma_{\max}]$.
It is possible, though uncommon, for line~\ref{line:asvd} of Algorithm~\ref{algo:vampSVD} to return a negative $\alpha_k$, which will lead to negative $\gamma_k$ and $\gammatil_k$ if not accounted for.
For the numerical results in Section~\ref{sec:num}, we used $\gamma_{\min}=1\times 10^{-11}$ and $\gamma_{\max}=1\times 10^{11}$.

Second, we find that a small amount of damping can be helpful when $\Abf$ is highly ill-conditioned.
In particular, we suggest to replace lines~\ref{line:xsvd} and \ref{line:gamsvd} of Algorithm~\ref{algo:vampSVD} with the damped versions
\begin{align}
\xbfhat_k
&= \rho \gbf_1(\rbf_k,\gamma_k) + (1-\rho) \xbfhat_{\km1}
\label{eq:xdamp}\\
\gamma_{\kp1}
&= \rho \gammatil_{k}\bkt{\dbf_k}R/(N-\bkt{\dbf_k}R) + (1-\rho) \gamma_k
\label{eq:gamdamp}
\end{align}
for all iterations $k>1$, where $\rho\in(0,1]$ is a suitably chosen damping parameter.
Note that, when $\rho=1$, the damping has no effect.
For the numerical results in Section~\ref{sec:num}, we used $\rho=0.97$.

Third, rather than requiring VAMP to complete $\Kit$ iterations, we suggest that the iterations are stopped when the normalized difference $\|\rbf_{1k}-\rbf_{1,\km1}\|/\|\rbf_{1k}\|$ falls below a tolerance $\tau$.
For the numerical results in Section~\ref{sec:num}, we used $\tau=1\times 10^{-4}$.

\textb{We note that the three minor modifications described above are standard features of many AMP implementations, such as the one in the GAMPmatlab toolbox \cite{GAMP-code}.  However, as discussed in Section~\ref{sec:alternatives}, they are not enough to stabilize AMP for in the case of ill-conditioned or non-zero-mean $\Abf$.}

Finally, we note that the VAMP algorithm requires the user to choose
the measurement-noise precision $\gamma_w$ and
the denoiser $\gbf_1(\cdot,\gamma_k)$.
Ideally, the true noise precision $\gamma_{w0}$ is known and the signal $\xbf^0$ is i.i.d.\ with known prior $p(x_j)$, in which case the MMSE denoiser can be straightforwardly designed.
In practice, however, $\gamma_{w0}$ and $p(x_j)$ are usually unknown.
Fortunately, there is a simple expectation-maximization (EM)-based method to estimate both quantities on-line, whose details are given in \cite{fletcher2016emvamp}.
The numerical results in \cite{fletcher2016emvamp} show that the convergence and asymptotic performance of EM-VAMP is nearly identical to that of VAMP with known $\gamma_{w0}$ and $p(x_j)$.
For the numerical results in Section~\ref{sec:num}, however, we assume that $\gamma_{w0}$ and $p(x_j)$ are known.

Matlab implementations of VAMP and EM-VAMP can be found in the public-domain GAMPmatlab toolbox \cite{GAMP-code}.

\section{State Evolution} \label{sec:SE}

\subsection{Large-System Analysis} \label{sec:large}
Our primary goal is to understand the behavior of the VAMP algorithm for a certain class of matrices in the high-dimensional regime.
We begin with an overview of our analysis framework and follow with more details in later sections.

\subsubsection{Linear measurement model}
Our analysis considers a sequence of problems indexed by the signal
dimension $N$.  For each $N$, we assume that there is a ``true" vector $\xbf^0\in\R^N$
which is observed through measurements of the form,
\beq \label{eq:yAxslr}
    \ybf = \Abf\xbf^0 + \wbf \in \R^N, \quad \wbf \sim \Norm(\mathbf{0}, \gamma_{w0}^{-1}\Ibf_N),
\eeq
where $\Abf\in\R^{N\times N}$ is a known transform and $\wbf$ is Gaussian noise with precision $\gamma_{w0}$.
Note that we use $\gamma_{w0}$ to denote the ``true" noise precision
to distinguish it from $\gamma_w$, which is the noise precision postulated by the estimator.

For the transform $\Abf$,
our key assumption is that it can be modeled as a large, \emph{right-orthogonally invariant} random matrix.
Specifically, we assume that it has an SVD of the form
\beq \label{eq:ASVD}
    \Abf=\Ubf\Sbf\Vbf\tran, \quad \Sbf = \Diag(\sbf),
\eeq
where $\Ubf$ and $\Vbf$ are $N\times N$ orthogonal matrices
such that $\Ubf$ is deterministic and
$\Vbf$ is Haar distributed (i.e.\ uniformly distributed on the set of orthogonal matrices).
We refer to $\Abf$ as ``right-orthogonally invariant'' because the distribution of $\Abf$ is identical to that of $\Abf\Vbf_0$ for any fixed orthogonal matrix $\Vbf_0$.
We will discuss the distribution of the singular values $\sbf\in\R^N$ below.

Although we have assumed that $\Abf$ is square to streamline the analysis,
we make this assumption without loss of generality.
For example, by setting
$$
\Ubf=\begin{bmatrix}\Ubf_0&\zero\\\zero&\Ibf\end{bmatrix},
\quad
\sbf=\begin{bmatrix}\sbf_0\\\zero\end{bmatrix},
$$
our formulation can model a wide rectangular matrix whose SVD is $\Ubf_0\Sbf_0\Vbf\tran$ with $\diag(\Sbf_0)=\sbf_0$.
A similar manipulation allows us to model a tall rectangular matrix.

\subsubsection{Denoiser}
Our analysis applies to a fairly general class of denoising
functions $\gbf_1(\cdot,\gamma_{1k})$ indexed by the parameter $\gamma_{1k}\geq 0$.  Our main assumption is that the denoiser is separable,
meaning that it is of the form \eqref{eq:g1sep} for some scalar denoiser $g_1(\cdot,\gamma_{1k})$.
As discussed above, this separability assumption will occur for the MAP and MMSE denoisers
under the assumption of an i.i.d.\ prior.  However, we do not require the denoiser to be
MAP or MMSE for any particular prior.
We will impose certain Lipschitz continuity conditions on $g_1(\cdot,\textb{\gamma_{1k}})$ in the sequel.

\subsubsection{Asymptotic distributions}
It remains to describe the distributions of the
true vector $\xbf^0$ and the singular-value vector $\sbf$.
A simple model would be to assume that they
are random i.i.d.\ sequences that grow with $N$.  However, following the
Bayati-Montanari analysis \cite{BayatiM:11},
we will consider a more general framework where each of these vectors is modeled as deterministic sequence
for which the empirical distribution of the components converges in distribution.
When the vectors $\xbf^0$ and $\sbf$ are i.i.d.\ random sequences, they will satisfy this condition almost surely.
Details of this analysis framework are reviewed in Appendix~\ref{sec:empConv}.

Using the definitions in Appendix~\ref{sec:empConv},
we assume that the components of the singular-value vector $\sbf\in\R^N$ in \eqref{eq:ASVD}
converge empirically with second-order moments as
\beq \label{eq:Slim}
    \lim_{N \arr \infty} \{ s_n \}_{n=1}^N \PLeq S,
\eeq
for some positive random variable $S$.  We assume that $\Exp[S] > 0$ and $S \in [0,S_{max}]$
for some finite maximum value $S_{max}$.
Additionally, we assume that
the components of the true vector, $\xbf^0$, and the initial input to the denoiser, $\rbf_{10}$,
converge empirically as
\beq \label{eq:RX0lim}
    \lim_{N \arr \infty} \{ (r_{10,n}, x^0_n) \}_{n=1}^N \PLeq (R_{10},X^0),
\eeq
for some random variables $(R_{10},X^0)$.
Note that the convergence with second-order moments
requires that $\Exp [(X^0)^2] < \infty$ and $\Exp [R^2_{10}] < \infty$, so they have bounded
second moments.
We also assume that the initial second-order term, \textb{if dependent on $N$,} converges as
\beq \label{eq:gam10lim}
    \lim_{N \arr \infty}\gamma_{10}\textb{(N)} = \gammabar_{10},
\eeq
for some $\gammabar_{10} > 0$.

As stated above, most of our analysis will apply to general separable denoisers
$g_1(\cdot,\gamma_{1k})$.  However, some results will apply specifically to MMSE denoisers.
Under the assumption that the components of the true vector $\xbf^0$ are asymptotically
distributed like the random variable $X^0$\textb{, as in \eqref{eq:RX0lim}}, the MMSE denoiser~\eqref{eq:g1mmsesca}
and its derivative \eqref{eq:g1dervar} reduce to
\begin{align} \label{eq:g1mmsex0}
\begin{split}
    g_1(r_1,\gamma_1) &= \Exp\left[ X^0 | R_1=r_1 \right], \\
    g_1'(r_1,\gamma_1) &= \gamma_1\var\left[ X^0 | R_1=r_1 \right],
\end{split}
\end{align}
where $R_1$ is the random variable representing $X^0$ corrupted by AWGN noise,
i.e.,
\[
    R_1 = X^0 + P, \quad P \sim \Norm(0,\gamma_1^{-1}),
\]
with $P$ being independent of $X^0$.  Thus, the MMSE denoiser and its derivative
can be computed from the
posterior mean and variance of $X^0$ under an AWGN measurement.

\subsection{Error Functions}

Before describing the state evolution (SE) equations and the
analysis in the LSL, we need to introduce two
key functions:  \emph{error functions} and \emph{sensitivity functions}.
We begin by describing the error functions.

The error functions, in essence,
describe the mean squared error (MSE)
of the denoiser and LMMSE estimators under AWGN measurements.
\textb{Recall from Section~\ref{sec:large}, that we have assumed that
the denoiser $\gbf_1(\cdot,\gamma_1)$ is separable with some componentwise function
$g_1(\cdot,\gamma_1)$.}
For this function $g_1(\cdot,\gamma_1)$, define the error function as
\begin{align}
    \MoveEqLeft \Ecal_1(\gamma_1,\tau_1)
    := \Exp\left[ (g_1(R_1,\gamma_1)-X^0)^2 \right],  \nonumber \\
    & R_1 = X^0 + P, \quad P \sim \Norm(0,\tau_1). \label{eq:eps1}
\end{align}
The function $\Ecal_1(\gamma_1,\tau_1)$ thus represents the MSE of the
estimate $\Xhat = g_1(R_1,\gamma_1)$ from a measurement $R_1$
corrupted by Gaussian noise of variance $\tau_1$.
For the LMMSE estimator, we define the error function as
\begin{align}
    \MoveEqLeft \Ecal_2(\gamma_2,\tau_2)
    := \lim_{N \arr \infty}
        \frac{1}{N} \Exp \left[ \| \gbf_2(\rbf_2,\gamma_2) -\xbf^0 \|^2 \right], \nonumber \\
    & \rbf_2 = \xbf^0 + \qbf, \quad \qbf \sim \Norm(0,\tau_2 \Ibf), \nonumber \\
    & \ybf = \Abf\xbf^0 + \wbf, \quad \wbf \sim \Norm(0,\gamma_{w0}^{-1} \Ibf),
    \label{eq:eps2}
\end{align}
which is the average per component error of the vector estimate under Gaussian noise.
Note that $\Ecal_2(\gamma_2,\tau_2)$ is implicitly a function of the noise precision levels $\gamma_{w0}$ and $\gamma_w$ \textb{(through $\gbf_2$ from \eqref{eq:g2slr})}, but this dependence is omitted to simplify the notation.

We will say that both estimators are ``matched" when
\[
    \tau_1 = \gamma_1^{-1}, \quad \tau_2 = \gamma_2^{-1}, \quad \gamma_w = \gamma_{w0},
\]
so that the noise levels used by the estimators both match the true noise levels.
Under the matched condition, we will use the simplified notation
\[
    \Ecal_1(\gamma_1) := \Ecal_1(\gamma_1,\gamma_1^{-1}), \quad
    \Ecal_2(\gamma_2) := \Ecal_2(\gamma_2,\gamma_2^{-1}).
\]
The following lemma establishes some basic properties of the error functions.

\begin{lemma} \label{lem:errfn} Recall the error functions $\Ecal_1,\Ecal_2$ defined above.
\begin{enumerate}[(a)]
\item  For the MMSE denoiser \eqref{eq:g1mmsex0} under the matched condition
$\tau_1=\gamma_1^{-1}$, the error function  is the conditional variance
\beq \label{eq:E1match}
    \Ecal_1(\gamma_1) = \var\left[ X^0 | R_1 = X^0 \textb{+ P} \right],~\textb{P\sim\Norm(0,\gamma_1^{-1})} .
\eeq

\item The LMMSE error function is given by
\beq \label{eq:eps2Q}
    \Ecal_2(\gamma_2,\tau_2) = \lim_{N \arr \infty} \frac{1}{N} \Tr\left[ \Qbf^{-2}
        \tilde{\Qbf} \right] ,
\eeq
where $\Qbf$ and $\tilde{\Qbf}$ are the matrices
\beq \label{eq:QQtdef}
    \Qbf : = \gamma_w \Abf\tran\Abf + \gamma_2\Ibf, \quad
    \tilde{\Qbf} :=
        \frac{\gamma_w^2}{\gamma_{w0}}\Abf\tran\Abf + \tau_2\gamma_2^2\Ibf.
\eeq
Under the matched condition $\tau_2 = \gamma_2^{-1}$ and $\gamma_w = \gamma_{w0}$,
\beq \label{eq:eps2Qmatch}
    \Ecal_2(\gamma_2) = \lim_{N \arr \infty} \frac{1}{N} \Tr\left[ \Qbf^{-1} \right].
\eeq

\item The LMMSE error function is also given by
\beq \label{eq:eps2S}
    \Ecal_2(\gamma_2,\tau_2) = \Exp\left[ \frac{
    \gamma_w^2 S^2/\gamma_{w0} + \tau_2\gamma_2^2}{(\gamma_w S^2 + \gamma_2)^2}  \right],
\eeq
where $S$ is the random variable \eqref{eq:Slim} representing the distribution of the
singular values of $\Abf$.  For the matched condition
$\tau_2 = \gamma_2^{-1}$ and $\gamma_w = \gamma_{w0}$,
\beq \label{eq:eps2Smatch}
    \Ecal_2(\gamma_2) = \Exp\left[
        \frac{1}{\gamma_wS^2+ \gamma_2} \right].
\eeq
\end{enumerate}
\end{lemma}
\begin{proof} See Appendix~\ref{sec:errsenspf}.
\end{proof}

\subsection{Sensitivity Functions}
The sensitivity functions describe the expected divergence of the estimator.
For the denoiser, the sensitivity function is defined as
\begin{align}
    \MoveEqLeft A_1(\gamma_1,\tau_1)
    := \Exp\left[ g_1'(R_1,\gamma_1) \right],  \nonumber \\
    & R_1 = X^0 + P, \quad P \sim \Norm(0,\tau_1), \label{eq:sens1}
\end{align}
which is the average derivative under a Gaussian noise input.  For the
LMMSE estimator, the sensitivity is defined as
\begin{align}
    \MoveEqLeft A_2(\gamma_2)
    := \lim_{N \arr \infty}
        \frac{1}{N} \Tr\left[ \frac{\partial \gbf_2(\rbf_2,\gamma_2)}{\partial \rbf_2}
            \right].
\end{align}

\begin{lemma} \label{lem:sens}
For the sensitivity functions above:
\begin{enumerate}[(a)]
\item  For the MMSE denoiser \eqref{eq:g1mmsex0} under the matched condition
$\tau_1=\gamma_1^{-1}$, the sensitivity function  is given by
\beq \label{eq:A1match}
    A_1(\gamma_1,\gamma_1^{-1}) = \gamma_1\var\left[ X^0 | R_1 = X^0 + \Norm(0,\gamma_1^{-1}) \right],
\eeq
which is the ratio of the conditional variance to the measurement variance $\gamma_1^{-1}$.

\item The LMMSE estimator's sensitivity function is given by
\[
    A_2(\gamma_2) = \lim_{N \arr \infty} \frac{1}{N}
        \gamma_2\Tr\left[ (\gamma_w \Abf\tran\Abf + \gamma_2\Ibf)^{-1} \right].
\]

\item The LMMSE estimator's sensitivity function can also be written as
\[
    A_2(\gamma_2) = \Exp\left[ \frac{\gamma_2}{\gamma_wS^2+ \gamma_2} \right].
\]
\end{enumerate}
\end{lemma}
\begin{proof} See Appendix~\ref{sec:errsenspf}.
\end{proof}

\subsection{State Evolution Equations}

We can now describe our main result, which is the SE equations
for VAMP.
For a given iteration $k \geq 1$, consider the set of components,
\[
     \{ (\xhat_{1k,n},r_{1k,n},x^0_n), ~ n=1,\ldots,N \}.
\]
This set represents the components of the true vector $\xbf^0$,
its corresponding estimate $\xbfhat_{1k}$ and the denoiser input
$\rbf_{1k}$.  \textb{Theorem~\ref{thm:se}} below will 
show that, under certain assumptions,
these components converge empirically as
\beq \label{eq:limrx1}
    \lim_{N \arr \infty} \{ (\xhat_{1k,n},r_{1k,n},x^0_n) \}
    \PLeq (\Xhat_{1k},R_{1k},X^0),
\eeq
where the random variables $(\Xhat_{1k},R_{1k},X^0)$ are given by
\begin{subequations} \label{eq:RX0var}
\begin{align}
    R_{1k} &= X^0 + P_k, \quad P_k \sim \Norm(0,\tau_{1k}), \\
    \Xhat_{1k} &= g_1(R_{1k},\gammabar_{1k}),
\end{align}
\end{subequations}
for constants $\gammabar_{1k}$ and $\tau_{1k}$ that will be defined below.
Thus, each component $r_{1k,n}$ appears as the true component $x^0_n$ plus
Gaussian noise.  The corresponding estimate $\xhat_{1k,n}$ then appears as the
denoiser output with $r_{1k,n}$ as the input.  Hence, the asymptotic behavior
of any component $x^0_n$ and its corresponding $\xhat_{1k,n}$ is identical to
a simple scalar system.  We will refer to \eqref{eq:limrx1}-\eqref{eq:RX0var} as the denoiser's \emph{scalar equivalent model}.

For the LMMSE estimation function,
we define the transformed error and transformed noise,
\beq \label{eq:qerrdef}
    \qbf_k := \Vbf\tran(\rbf_{2k}-\xbf^0), \quad \xibf := \Ubf\tran\wbf,
\eeq
where $\Ubf$ and $\Vbf$ are the matrices in the SVD decomposition \eqref{eq:ASVD}.
\textb{Theorem~\ref{thm:se} will also show that}  
these transformed errors and singular values $s_n$ converge as,
\beq \label{eq:limqxi}
    \lim_{N \arr \infty} \{ (q_{k,n},\xi_n,s_n) \}
    \PLeq (Q_k,\Xi,S),
\eeq
to a set of random variables $(Q_k,\Xi,S)$.
These random variables are independent, with
$S$ defined in the limit \eqref{eq:Slim} and
\beq \label{eq:QXivar}
    Q_k \sim \Norm(0,\tau_{2k}), \quad \Xi \sim \Norm(0,\gamma_{w0}^{-1}),
\eeq
where  $\tau_{2k}$ is a variance that will be defined below and $\gamma_{w0}$
is the noise precision in the measurement model \eqref{eq:yAxslr}.
Thus \eqref{eq:limqxi}-\eqref{eq:QXivar} is a scalar equivalent model for the LMMSE estimator.

The variance terms are defined recursively through what are called \emph{state evolution}
equations,
\begin{subequations} \label{eq:se}
\begin{align}
    \alphabar_{1k} &= A_1(\gammabar_{1k},\tau_{1k}) \label{eq:a1se} \\
    \etabar_{1k} &= \frac{\gammabar_{1k}}{\alphabar_{1k}}, \quad
    \gammabar_{2k} = \etabar_{1k} - \gammabar_{1k} \label{eq:eta1se} \\\
    \tau_{2k} &= \frac{1}{(1-\alphabar_{1k})^2}\left[
        \Ecal_1(\gammabar_{1k},\tau_{1k}) - \alphabar_{1k}^2\tau_{1k} \right],
            \label{eq:tau2se} \\
    \alphabar_{2k} &= A_2(\gammabar_{2k},\tau_{2k}) \label{eq:a2se} \\
    \etabar_{2k} &= \frac{\gammabar_{2k}}{\alphabar_{2k}}, \quad
    \gammabar_{1,\kp1} = \etabar_{2k} - \gammabar_{2k} \label{eq:eta2se} \\
    \tau_{1,\kp1} &= \frac{1}{(1-\alphabar_{2k})^2}\left[
        \Ecal_2(\gammabar_{2k},\tau_{2k}) - \alphabar_{2k}^2\tau_{2k} \right],
            \label{eq:tau1se}
\end{align}
\end{subequations}
which are initialized with
\beq
    \tau_{10} = \Exp[(R_{10}-X^0)^2],
\eeq
and $\gammabar_{10}$ defined from the limit \eqref{eq:gam10lim}.

\begin{theorem} \label{thm:se}
Under the above assumptions and definitions, assume additionally that for all iterations $k$:
\begin{enumerate}[(i)]
\item The solution $\alphabar_{1k}$ from the SE equations \eqref{eq:se} satisfies
\beq \label{eq:asecon}
    \alphabar_{1k} \in (0,1).
\eeq
\item The functions $A_i(\gamma_i,\tau_i)$ and $\Ecal_i(\gamma_i,\tau_i)$
are continuous at $(\gamma_i,\tau_i)=(\gammabar_{ik},\tau_{ik})$.
\item The denoiser function $g_1(r_1,\gamma_1)$ and its derivative
 $g_1'(r_1,\gamma_1)$
are uniformly Lipschitz in $r_1$ at $\gamma_1=\gammabar_{1k}$.
(See  Appendix~\ref{sec:empConv} for a precise definition of uniform Lipschitz continuity.)
\end{enumerate}
Then, for any fixed iteration $k \geq 0$,
\beq \label{eq:aglim}
    \lim_{N \arr \infty} (\alpha_{ik},\eta_{ik},\gamma_{ik}) =
    (\alphabar_{ik},\etabar_{ik}, \gammabar_{ik})
\eeq
almost surely.
In addition, the empirical limit \eqref{eq:limrx1} holds almost surely for all $k > 0$,
and \eqref{eq:limqxi} holds almost surely for all $k \geq 0$.
\end{theorem}

\subsection{Mean Squared Error}
One important
use of the scalar equivalent model is to predict the asymptotic performance
of the VAMP algorithm in the LSL.  For example, define the asymptotic
mean squared error (MSE) of the iteration-$k$ estimate $\xbfhat_{ik}$ as
\beq \label{eq:mseLim}
    \MSE_{ik} := \lim_{N \arr \infty} \frac{1}{N}\|\xbfhat_{ik}-\xbf^0\|^2.
\eeq
For this MSE, we claim that
\beq \label{eq:mseEcal}
    \MSE_{ik} = \Ecal_i(\gammabar_{ik},\tau_{ik}).
\eeq
To prove \eqref{eq:mseEcal} for $i=1$, we write
\begin{align*}
    \MSE_{1k}
    &=\lim_{N \arr \infty}
        \frac{1}{N} \sum_{n=1}^N (\xhat_{1k,n}-x^0_n)^2     \\
    &\stackrel{(a)}{=} \Exp[(\Xhat_{1k}-X^0)^2] \\
    &\stackrel{(b)}{=} \Exp[(g_1(R_1,\gammabar_{1k})-X^0)^2]
     \stackrel{(c)}{=} \Ecal_1(\gammabar_{1k},\tau_{1k})
\end{align*}
where (a) and (b) follow from the convergence in \eqref{eq:limrx1} and the scalar equivalent model \eqref{eq:limrx1},
and where (c) follows from \eqref{eq:eps1}.  Using the scalar equivalent model~\eqref{eq:limqxi},
the definition of $\Ecal_2(\cdot)$ in \eqref{eq:eps2}, and calculations similar to the proof
of Lemma~\ref{lem:errfn}, one can also show that \eqref{eq:mseEcal} holds for $i=2$.

Interestingly, this type of calculation can be used to compute any other componentwise distortion metric.
Specifically, given any distortion function $d(x,\xhat)$ that is pseudo-Lipschitz of order two,
its average value is given by
\[
    \lim_{N \arr \infty}
        \frac{1}{N} \sum_{n=1}^N d(x^0_n,\xhat_{1k,n}) = \Exp\left[ d(X^0,\Xhat_{1k}) \right],
\]
where the expectation is from  the scalar equivalent model \eqref{eq:limrx1}.

\subsection{Contractiveness of the Denoiser}
An essential requirement of Theorem~\ref{thm:se} is the condition~\eqref{eq:asecon}
that $\alphabar_{1k} \in (0,1)$.  This assumption requires that, in a certain average,
the denoiser function $g_1(\cdot,\gamma_1)$ is increasing (i.e., $g_1'(r_{1n},\gamma_1) > 0$)
and is a contraction (i.e., $g_1'(r_{1n},\gamma_1) < 1$).
If these conditions are not met,
then $\alphabar_{1k} \leq 0$ or $\alphabar_{1k} \geq 1$,
and either the estimated precision $\etabar_{1k}$ or $\gammabar_{2k}$
in \eqref{eq:eta1se} may be negative, causing subsequent updates to be invalid.
Thus, $\alphabar_{1k}$ must be in the range $(0,1)$.
There are two important conditions under which
this increasing contraction property are provably guaranteed:

\noindent
\paragraph*{Strongly convex penalties}  Suppose that  $g_1(r_{1n},\gamma_1)$ is the
either the MAP denoiser \eqref{eq:gmapsca} or the MMSE denoiser
\eqref{eq:g1mmsesca} for a density $p(x_n)$ that is strongly log-concave.
That is, there exists constants $c_1,c_2 > 0$ such that
\[
    c_1 \leq -\frac{\partial^2}{\partial x_n^2} \ln p(x_n) \leq c_2.
\]
Then, using results from log-concave functions \cite{brascamp2002extensions},
it is shown in \cite{rangan2015admm} that
\[
    g_1'(r_{1n},\gamma_1) \in
        \left[ \frac{\gamma_1}{c_2 + \gamma_1}, \frac{\gamma_1}{c_1 + \gamma_1}\right]
        \subset (0,1),
\]
for all $r_{1n}$ and $\gamma_1 > 0$.  Hence, from the definition of the
sensitivity function \eqref{eq:sens1},
the sensitivity $\alphabar_{1k}$ in \eqref{eq:a1se} will be in the range $(0,1)$.

\noindent
\paragraph*{Matched MMSE denoising}  Suppose that  $g_1(r_{1n},\gamma_1)$ is the
MMSE denoiser in the matched condition where $\gammabar_{1k}=\tau_{1k}^{-1}$
for some iteration $k$.
From \eqref{eq:A1match},
\[
    A_1(\gamma_1,\gamma_1^{-1}) = \gamma_1\var\left[ X^0 | R_1 = X^0 + \Norm(0,\gamma_1^{-1})
        \right].
\]
Since the conditional variance is positive, $A_1(\gamma_1,\gamma_1^{-1}) > 0$.
Also, since the variance is bounded above by the MSE of a linear estimator,
\begin{align*}
    \MoveEqLeft \gamma_1\var\left[ X^0 | R_1 = X^0 + \Norm(0,\gamma_1^{-1})\right]\\
    &\leq \gamma_1\frac{\gamma_1^{-1}\tau_{x_0}}{\tau_{x_0}+\gamma_1^{-1}}
    = \frac{\gamma_1\tau_{x_0}}{1+ \gamma_1\tau_{x_0}} < 1,
\end{align*}
where $\tau_{x0} = \var(X^0)$.
Thus, we have $A_1(\gamma_1,\gamma_1^{-1}) \in (0,1)$ and
$\alphabar_{1k} \in (0,1)$.

\medskip
In the case when the prior is not log-concave and the estimator uses an denoiser
that is not perfectly matched, $\alphabar_{1k}$ may not be in the valid range $(0,1)$.
In these cases, VAMP may obtain invalid (i.e.\ negative) variance estimates.

\section{MMSE Denoising, Optimality, and Connections to the Replica Method} \label{sec:replica}

An important special case of the VAMP algorithm
is when we apply the MMSE optimal denoiser under matched $\gamma_w$.  In this case,
the SE equations simplify considerably.

\begin{theorem} \label{thm:seMmse}  Consider the SE equations \eqref{eq:se} with
the MMSE optimal denoiser
\eqref{eq:g1mmsex0}, matched $\gamma_w=\gamma_{w0}$, and matched initial condition $\gammabar_{10} = \tau_{10}^{-1}$.
Then, for all iterations $k \geq 0$,
\begin{subequations} \label{eq:sematch}
\begin{align}
    \etabar_{1k} &= \frac{1}{\Ecal_1(\gammabar_{1k})}, \quad
    \gammabar_{2k} = \tau_{2k}^{-1} = \etabar_{1k} - \gammabar_{1k},
        \label{eq:eta1sematch} \\
    \etabar_{2k} &= \frac{1}{\Ecal_2(\gammabar_{2k})}, \quad
    \gammabar_{1,\kp1} = \tau_{1,\kp1}^{-1} = \etabar_{2k} - \gammabar_{2k}.
        \label{eq:eta2sematch}
\end{align}
In addition, for estimators $i=1,2$,  $\etabar_{ik}$ is the inverse MSE:
\beq \label{eq:etammse}
    \etabar_{ik}^{-1} = \lim_{N \arr \infty} \frac{1}{N} \|\xbfhat_{ik}-\xbf^0\|^2.
\eeq
\end{subequations}
\end{theorem}
\begin{proof}  See Appendix~\ref{sec:seMmsePf}.
\end{proof}

It is useful to compare this result with
the work \cite{tulino2013support}, which uses the \emph{replica method} from statistical physics
to predict the asymptotic MMSE error in the LSL.  To state the result,
given a positive semidefinite matrix $\Cbf$, we define its Stieltjes transform as
\beq \label{eq:stieltjes}
    S_{\Cbf}(\omega) = \frac{1}{N} \Tr\left[ (\Cbf - \omega \Ibf_N)^{-1} \right]
     = \frac{1}{N} \sum_{n=1}^N \frac{1}{\lambda_n - \omega},
\eeq
where $\lambda_n$ are the eigenvalues of $\Cbf$. Also, let $R_{\Cbf}(\omega)$
denote the so-called $R$-transform of $\Cbf$, given by
\beq \label{eq:rtrans}
    R_{\Cbf}(\omega) =  S_{\Cbf}^{-1}(-\omega) - \frac{1}{\omega},
\eeq
where the inverse $S_{\Cbf}^{-1}(\cdot)$ is in terms of composition of functions.
The Stieltjes and $R$-transforms are discussed in detail in \cite{TulinoV:04}.
The Stieltjes and $R$-transforms can be
extended to random matrix sequences by taking limits as $N \arr\infty$
(for matrix sequences where these limits converge almost surely).

Now suppose that $\xbfhat = \Exp[\xbf^0|\ybf]$ is the MMSE estimate of $\xbf^0$ given $\ybf$.
Let $\etabar^{-1}$ be the asymptotic inverse MSE
\[
    \etabar^{-1} := \lim_{N \arr \infty} \frac{1}{N} \|\xbfhat-\xbf^0\|^2.
\]
Using a so-called replica symmetric analysis, it is argued in
\cite{tulino2013support} that this MSE should satisfy the fixed point
equations
\beq \label{eq:fixrep}
    \gammabar_1 = R_{\Cbf}(-\etabar^{-1}), \quad
    \etabar^{-1} =  \Ecal_1( \gammabar_1 ),
\eeq
where $\Cbf=\gamma_{w0}\Abf\tran\Abf$.
A similar result is given in \cite{kabashima2014signal}.

\begin{theorem} \label{thm:replica}  Let $\gammabar_i,\etabar_i$ be any fixed point
solutions to the SE equations \eqref{eq:sematch} of VAMP under MMSE denoising and matched $\gamma_w=\gamma_{w0}$.
Then $\etabar_1=\etabar_2$.  If we define $\etabar:=\etabar_i$ as the common value,
then $\gammabar_1$ and $\etabar$ satisfy the replica fixed point equation \eqref{eq:fixrep}.
\end{theorem}
\begin{proof}  Note that we have dropped the iteration index $k$ since we are discussing
a fixed point.  First, \eqref{eq:sematch} shows that, at any fixed point,
\[
    \gammabar_1+\gammabar_2 = \etabar_1 = \etabar_2,
\]
so that $\etabar_1=\etabar_2$.  Also, in the matched case, \eqref{eq:eps2Smatch} shows that
\[
    \Ecal_2(\gammabar_2) = S_{\Cbf}(-\gammabar_2) .
\]
Since $\etabar^{-1} = \Ecal_2(\gammabar_2)$, we have that
\[
    \gammabar_1 = \etabar-\gammabar_2 = \etabar+S_\Cbf^{-1}(\etabar^{-1}) = R_\Cbf(-\etabar^{-1}).
\]
Also,  $\etabar^{-1}=\etabar_1^{-1} = \Ecal(\gammabar_1)$.
\end{proof}

The consequence of Theorem~\ref{thm:replica} is that, if the replica equations \eqref{eq:fixrep}
have a unique fixed point, then the MSE achieved by the VAMP algorithm
exactly matches the Bayes optimal MSE as predicted by the replica method.  Hence, if this
replica prediction is correct, then the VAMP method provides a computationally efficient
method for finding MSE optimal estimates under very general priors---including priors
for which the associated penalty functions are not convex.

The replica method, however, is generally heuristic.
But in the case of i.i.d.\ Gaussian matrices,
it has recently been proven that the replica prediction is correct
\cite{reeves2016replica,barbier2016mutual}.

\section{Numerical Experiments} \label{sec:num}

In this section, we present numerical experiments that compare
the VAMP\footnote{A Matlab implementation of VAMP can be found in the public-domain GAMPmatlab toolbox \cite{GAMP-code}.}
Algorithm~\ref{algo:vampSVD} to
the VAMP state evolution from Section~\ref{sec:SE},
the replica prediction from \cite{tulino2013support},
the AMP Algorithm~\ref{algo:amp} from \cite{DonohoMM:10-ITW1},
the S-AMP algorithm from \cite[Sec.~IV]{cakmak2014samp},
the adaptively damped (AD) GAMP algorithm from \cite{Vila:ICASSP:15},
and the support-oracle MMSE estimator, whose MSE lower bounds that achievable by any practical method.
In all cases, we consider the recovery of vectors $\xbf^0\in\R^N$
from AWGN-corrupted measurements $\ybf\in\R^M$ constructed from \eqref{eq:yAx}, where
$\xbf^0$ was drawn i.i.d.\ zero-mean Bernoulli-Gaussian with $\Pr\{x^0_j\neq 0\}=0.1$,
where $\wbf\sim\Norm(\zero,\Ibf/\gamma_{w0})$,
and where $M=512$ and $N=1024$.
All methods under test were matched to the true signal and noise statistics.
When computing the support-oracle MMSE estimate, the support of $\xbf^0$ is assumed to be known, in which case the problem reduces to estimating the non-zero coefficients of $\xbf^0$.
Since these non-zero coefficients are Gaussian, their MMSE estimate can be computed in closed form.
For VAMP we used the implementation enhancements described in Section~\ref{sec:implementation}.
For line~\ref{line:gamma} of AMP Algorithm~\ref{algo:amp}, we used $1/\gamma_{\kp1}=1/\gamma_{w0}+\frac{N}{M}\alpha_k/\gamma_k$, as specified in \cite[Eq.\ (25)]{DonohoMM:10-ITW1}.
For the AMP, S-AMP, and AD-GAMP algorithms, we allowed a maximum of $1000$ iterations, and for the VAMP algorithm we allowed a maximum of $100$ iterations.

\subsection{Ill-conditioned $\Abf$} \label{sec:ill}

First we investigate algorithm robustness to the condition number of $\Abf$.
For this study, realizations of $\Abf$ were constructed from the SVD $\Abf=\Ubfbar\Diag(\sbfbar)\Vbfbar\tran\in\R^{M\times N}$ with geometric singular values $\sbfbar\in\R^M$.
That is, $\bar{s}_i/\bar{s}_{i-1}=\rho~\forall i$, with $\rho$ chosen to achieve a desired condition number $\kappa(\Abf):=\bar{s}_1/\bar{s}_M$ and with $\bar{s}_1$ chosen so that $\|\Abf\|_F^2=N$.
The singular vector matrices $\Ubfbar,\Vbfbar$ were drawn uniformly at random from the group of orthogonal matrices, i.e., from the Haar distribution.
Finally, the signal and noise variances were set to achieve a signal-to-noise ratio (SNR) $\Exp[\|\Abf\xbf\|^2]/\Exp[\|\wbf\|^2]$ of $40$~dB.

Figure~\ref{fig:nmse_vs_cond} plots the median normalized MSE (NMSE) achieved by each algorithm over $500$ independent realizations of $\{\Abf,\xbf,\wbf\}$, where $\text{NMSE}(\xbfhat):=\|\xbfhat-\xbf^0\|^2/\|\xbf^0\|^2$.
To enhance visual clarity, NMSEs were clipped to a maximum value of $1$.
Also, error bars are shown that (separately) quantify the positive and negative standard deviations of VAMP's NMSE from the median value.
The NMSE was evaluated for condition numbers $\kappa(\Abf)$ ranging from $1$ (i.e., row-orthogonal $\Abf$) to $1\times10^6$ (i.e., highly ill-conditioned $\Abf$).

\begin{figure}[t]
\centering
\psfrag{SNR=40dB, N=1024, M=512, rho=0.2, U=Haar, V=Haar, isCmplx=0, median of 500}{}
\psfrag{condition number}[t][t][0.7]{\sf condition number $\kappa(\Abf)$}
\psfrag{median NMSE [dB]}[b][b][0.7]{\sf median NMSE [dB]}
\psfrag{damped GAMP}[lB][lB][0.42]{\sf AD-GAMP}
\includegraphics[width=\figsize]{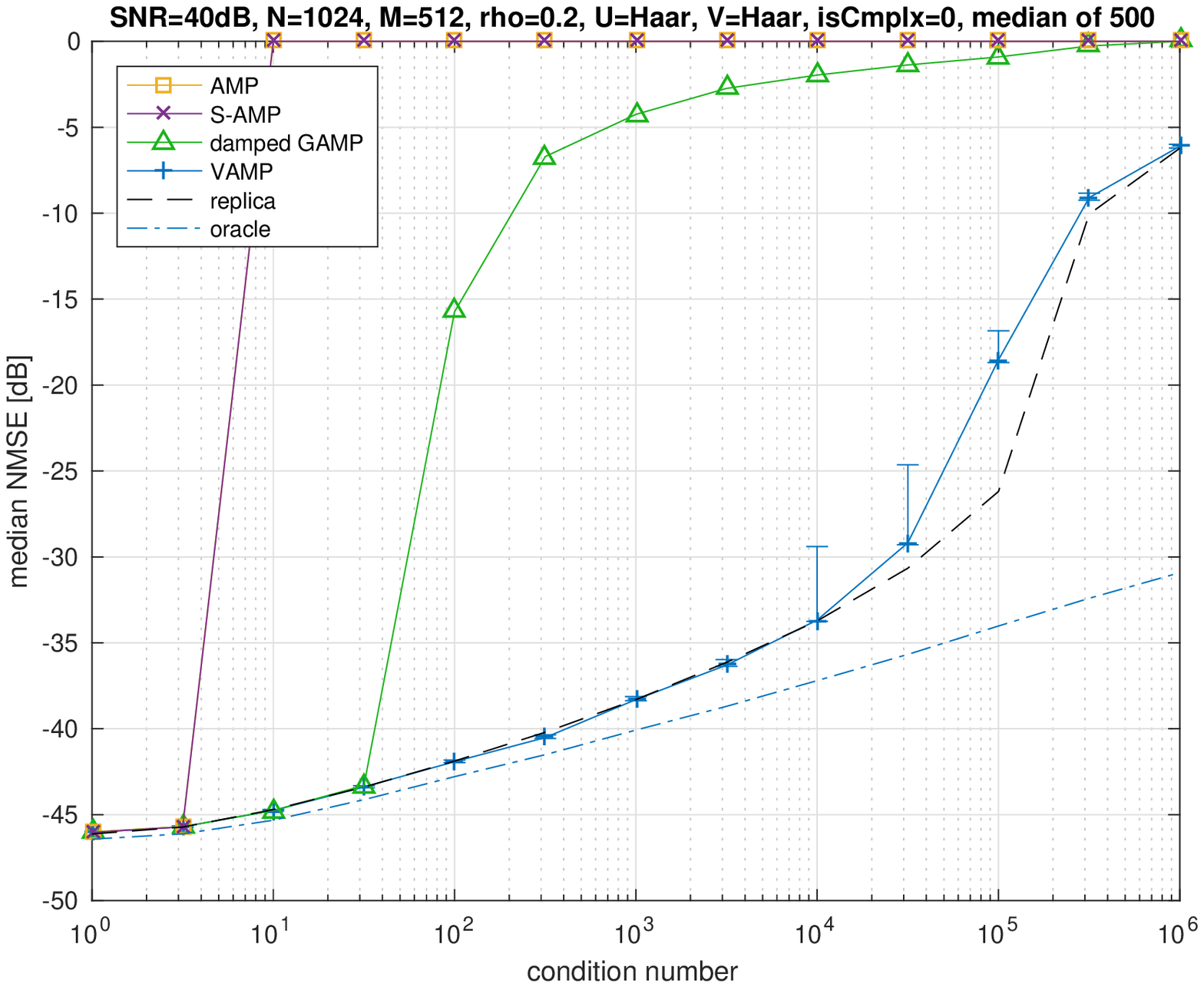}
\caption{NMSE versus condition number $\kappa(\Abf)$ at final algorithm iteration.  The reported NMSE is the median over $500$ realizations, with error bars shown on the VAMP trace.
\label{fig:nmse_vs_cond}}
\end{figure}

In Figure~\ref{fig:nmse_vs_cond}, we see that AMP and S-AMP diverged for even mildly ill-conditioned $\Abf$.
We also see that, while adaptive damping helped to extend the operating range of AMP, it had a limited effect.
In contrast, Figure~\ref{fig:nmse_vs_cond} shows that VAMP's NMSE stayed relatively close to the replica prediction for all condition numbers $\kappa(\Abf)$.
The small gap between VAMP and the replica prediction is due to finite-dimensional effects; the SE analysis from Section~\ref{sec:SE} establishes that this gap closes in the large-system limit.
Finally, Figure~\ref{fig:nmse_vs_cond} shows that the oracle bound is close to the replica prediction at small $\kappa(\Abf)$ but not at large $\kappa(\Abf)$.

Figure~\ref{fig:nmse_vs_iter_cond}(a) plots NMSE versus algorithm iteration for condition number $\kappa(\Abf)=1$ and Figure~\ref{fig:nmse_vs_iter_cond}(b) plots the same for $\kappa(\Abf)=1000$, again with error bars on the VAMP traces.
Both figures show that the VAMP trajectory stayed very close to the VAMP-SE trajectory at every iteration.
The figures also show that VAMP converges a bit quicker than AMP, S-AMP, and AD-GAMP when $\kappa(\Abf)=1$, and that VAMP's convergence rate is relatively insensitive to the condition number $\kappa(\Abf)$.

\begin{figure}[t]
\centering
\psfrag{condition number=1}[b][b][0.7]{(a)}
\psfrag{condition number=1000}[b][b][0.7]{(b)}
\psfrag{iterations}[t][t][0.7]{\sf iteration}
\psfrag{median NMSE [dB]}[b][b][0.7]{\sf median NMSE [dB]}
\psfrag{damped GAMP}[lB][lB][0.42]{\sf AD-GAMP}
\includegraphics[width=\figsize]{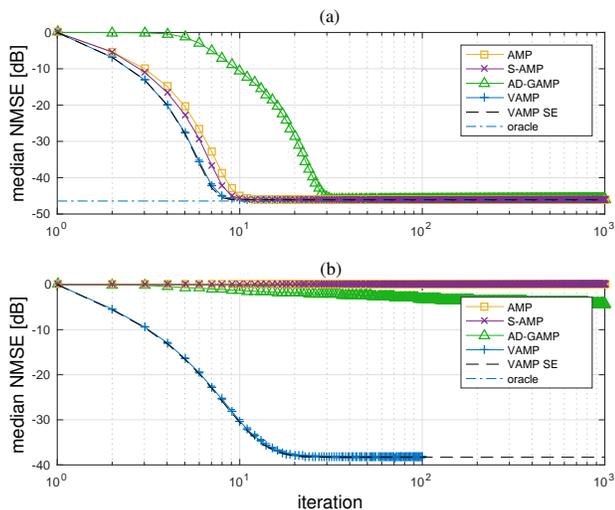}
\caption{NMSE versus algorithm iteration for condition number $\kappa(\Abf)=1$ in (a) and $\kappa(\Abf)=1000$ in (b).  The reported NMSE is the median over $500$ realizations, with error bars shown on the VAMP traces.
\label{fig:nmse_vs_iter_cond}}
\end{figure}

\subsection{Non-zero-mean $\Abf$} \label{sec:nzmean}

In this section, we investigate algorithm robustness to the componentwise mean of $\Abf$.
For this study, realizations of $\Abf$ were constructed by first drawing an i.i.d.\ $\Norm(\mu,1/M)$ matrix and then scaling it so that $\|\Abf\|_F^2=N$ (noting that essentially no scaling is needed when $\mu\approx 0$).
As before, the signal and noise variances were set to achieve an SNR of $40$~dB.
For AD-GAMP, we used the mean-removal trick proposed in \cite{Vila:ICASSP:15}.

Figure~\ref{fig:nmse_vs_mean} plots the NMSE achieved by each algorithm over $200$ independent realizations of $\{\Abf,\xbf,\wbf\}$.
The NMSE was evaluated for mean parameters $\mu$ between $0.001$ and $10$.
Note that, when $\mu>0.044$, the mean is larger than the standard deviation.
Thus, the values of $\mu$ that we consider are quite extreme relative to past studies like \cite{Caltagirone:14-ISIT}.

\begin{figure}[t]
\centering
\psfrag{SNR=40dB, N=1024, M=512, rho=0.2, isCmplx=0, median of 200}{}
\psfrag{nzmean}[t][t][0.7]{\sf mean $\mu$ of $\Abf$}
\psfrag{median NMSE [dB]}[b][b][0.7]{\sf median NMSE [dB]}
\psfrag{damped GAMP}[lB][lB][0.42]{\sf MAD-GAMP}
\includegraphics[width=\figsize]{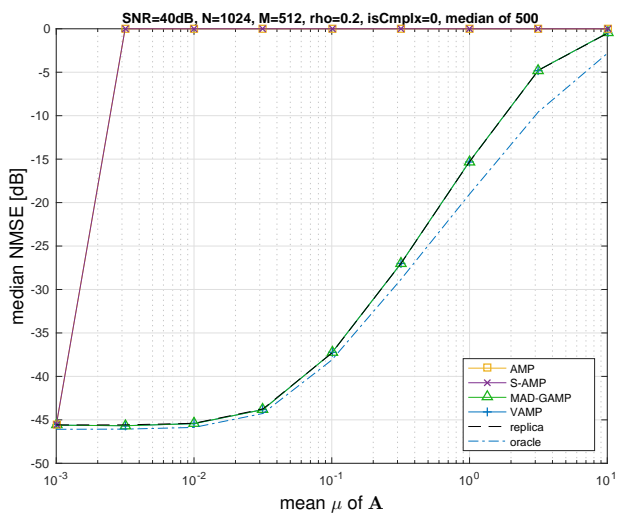}
\caption{NMSE versus mean $\mu$ at final algorithm iteration.  The reported NMSE is the median over $200$ realizations, with error bars shown on the VAMP trace.
\label{fig:nmse_vs_mean}}
\end{figure}

Figure~\ref{fig:nmse_vs_mean} shows that AMP and S-AMP diverged for even mildly mean-perturbed $\Abf$.
In contrast, the figure shows that VAMP and mean-removed AD-GAMP (MAD-GAMP) closely matched the replica prediction for all mean parameters $\mu$.
It also shows a relatively small gap between the replica prediction and the oracle bound, especially for small $\mu$.

\begin{figure}[t]
\centering
\psfrag{nzmean=0.001}[b][b][0.7]{(a)}
\psfrag{nzmean=1}[b][b][0.7]{(b)}
\psfrag{iterations}[t][t][0.7]{\sf iteration}
\psfrag{median NMSE [dB]}[b][b][0.7]{\sf median NMSE [dB]}
\psfrag{damped GAMP}[lB][lB][0.42]{\sf MAD-GAMP}
\includegraphics[width=\figsize]{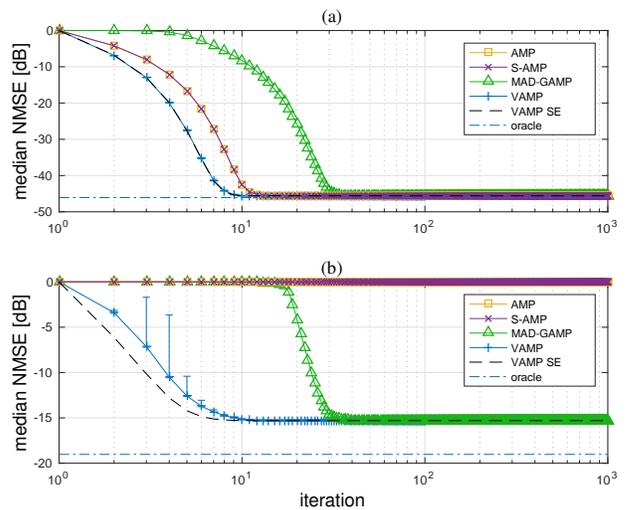}
\caption{NMSE versus algorithm iteration when $\Abf$ has mean $\mu=0.001$ in (a) and $\mu=1$ in (b).  The reported NMSE is the median over $200$ realizations, with error bars shown on the VAMP traces.
\label{fig:nmse_vs_iter_mean}}
\end{figure}

Figure~\ref{fig:nmse_vs_iter_mean}(a) plots NMSE versus algorithm iteration for matrix mean $\mu=0.001$ and Figure~\ref{fig:nmse_vs_iter_mean}(b) plots the same for $\mu=1$.
When $\mu=0.001$, VAMP closely matched its SE at all iterations and converged noticeably quicker than AMP, S-AMP, and MAD-VAMP.
When $\mu=1$, there was a small but noticeable gap between VAMP and its SE for the first few iterations, although the gap closed after about $10$ iterations.
This gap may be due to the fact that the random matrix $\Abf$ used for this experiment was not right-\textb{orthogonally} invariant, since the dominant singular vectors are close to (scaled versions of) the $\one$s vector for sufficiently large $\mu$.

\subsection{Row-orthogonal $\Abf$} \label{sec:SNR}

In this section we investigate algorithm NMSE versus SNR for row-orthogonal $\Abf$, i.e., $\Abf$ constructed as in Section~\ref{sec:ill} but with $\kappa(\Abf)=1$.
Previous studies \cite{cakmak2015samp,kabashima2014signal} have demonstrated that, when $\Abf$ is \textb{orthogonally} invariant but not i.i.d.\ Gaussian (e.g., row-orthogonal),
the fixed points of S-AMP and diagonal-restricted EC are better than those of AMP because the former approaches exploit the singular-value spectrum of $\Abf$, whereas AMP does not.

Table~\ref{tab:nmse_vs_snr} reports the NMSE achieved by VAMP, S-AMP, and AMP at three levels of SNR: $10$~dB, $20$~dB, and $30$~dB.
The NMSEs reported in the table were computed from an average of $1000$ independent realizations of $\{\Abf,\xbf,\wbf\}$.
Since the NMSE differences between the algorithms are quite small, the table also reports the standard error on each NMSE estimate to confirm its accuracy.

Table~\ref{tab:nmse_vs_snr} shows that VAMP and S-AMP gave nearly identical NMSE at all tested SNRs, which is expected because these two algorithms share the same fixed points.
The table also shows that VAMP's NMSE was strictly better than AMP's NMSE at low SNR (as expected), but that the NMSE difference narrows as the SNR increases.
Finally, the table reports the replica prediction of the NMSE, which is about $3\%$ lower (i.e., $-0.15$~dB) than VAMP's empirical NMSE at each SNR.
We attribute this difference to finite-dimensional effects.

\begin{table}[t]
\centering
\begin{tabular}{@{}c@{~} | @{~}c @{~~} r@{}l @{~~} r@{}l @{~~} r@{}l@{}}
SNR &replica  &VAMP&(stderr)      &S-AMP&(stderr)      &AMP&(stderr) \\\hline
10 dB&5.09e-02&5.27e-02&(4.3e-04)&5.27e-02&(4.3e-04)&5.42e-02&(4.2e-04)\\
20 dB&3.50e-03&3.57e-03&(2.7e-05)&3.58e-03&(2.7e-05)&3.62e-03&(2.6e-05)\\
30 dB&2.75e-04&2.84e-04&(2.2e-06)&2.85e-04&(2.2e-06)&2.85e-04&(2.1e-06)\\
\end{tabular}
\medskip

\caption{Average NMSE versus SNR for row-orthogonal $\Abf$, where the average was computed from $1000$ realizations. Standard error deviations are also reported.
\label{tab:nmse_vs_snr}}
\end{table}

\subsection{Discussion}

Our numerical results confirm what is already known about the \emph{fixed points} of diagonally restricted EC (via VAMP) and S-AMP.
That is,
when $\Abf$ is large and right-\textb{orthogonally} invariant,
they agree with each other and with the replica prediction;
and when $\Abf$ is large i.i.d.\ Gaussian (which is a special case of right-\textb{orthogonally} invariant \cite{TulinoV:04}),
they furthermore agree with the fixed points of AMP
\cite{cakmak2015samp,kabashima2014signal}.

But our numerical results also clarify that it is not enough for an algorithm to have good fixed points, because it may not converge to its fixed points.
For example, although the fixed points of S-AMP are good (i.e., replica matching) for \emph{any} large right-\textb{orthogonally} invariant $\Abf$, our numerical results indicate that S-AMP converges only for a small subset of large right-\textb{orthogonally} invariant $\Abf$: those with singular-value spectra similar (or flatter than) i.i.d.\ Gaussian $\Abf$.

The SE analysis from Section~\ref{sec:SE} establishes that, in the large-system limit and under matched priors, VAMP is guaranteed to converge to a fixed point that is also a fixed point of the replica equation \eqref{eq:fixrep}.
Our numerical results suggest that, even with large but finite-dimensional right \textb{orthogonally} invariant $\Abf$ (i.e., $512\times 1024$ in our simulations),
VAMP attains NMSEs that are very close to the replica prediction.

\section{Conclusions} \label{sec:conc}

In this paper, we considered the standard linear regression (SLR) problem \eqref{eq:yAx}, where the goal is to recover the vector $\xbf^0$ from noisy linear measurements $\ybf=\Abf\xbf^0+\wbf$.
Our work is inspired by Donoho, Maleki, and Montanari's AMP algorithm \cite{DonohoMM:09}, which offers a computationally efficient approach to SLR.
AMP has the desirable property that its behavior is rigorously characterized under large i.i.d.\ sub-Gaussian $\Abf$ by a scalar state evolution whose fixed points, when unique, are Bayes optimal \cite{BayatiM:11}.
A major shortcoming of AMP, however, is its fragility with respect to the i.i.d.\ sub-Gaussian model on $\Abf$: even small perturbations from this model can cause AMP to diverge.

In response, we proposed a vector AMP (VAMP) algorithm that (after performing an initial SVD) has similar complexity to AMP but is much more robust with respect to the matrix $\Abf$.
Our main contribution is establishing that VAMP's behavior can be rigorously characterized by a scalar state-evolution that holds for large, right-\textb{orthogonally} invariant $\Abf$.
The fixed points of VAMP's state evolution are, in fact,
consistent with the replica prediction of the minimum mean-squared
error recently derived in \cite{tulino2013support}.
We also showed how VAMP can be derived as an approximation of belief propagation on a factor graph with vector-valued nodes, hence the name ``vector AMP.''
Finally, we presented numerical experiments to demonstrate VAMP's robust convergence for ill-conditioned and mean-perturbed matrices $\Abf$ that cause earlier AMP algorithms to diverge.

As future work, it would be interesting to extend VAMP to the generalized linear model, where the outputs $\Abf\xbf^0$ are non-linearly mapped to $\ybf$.
Also, it would be interesting to design and analyze extensions of VAMP that are robust to \textb{more general models for $\Abf$, such as the case where $\Abf$ is statistically coupled to $\xbf^0$.}

\appendices

\section{Message-Passing Derivation of VAMP} \label{sec:EP}

In this appendix, we detail the message-passing derivation of Algorithm~\ref{algo:vamp}.
Below, we will use $k$ to denote the VAMP iteration and $n$ to index the elements of $N$-dimensional vectors like $\xbf_{1},\rbf_{1k}$ and $\xbfhat_{1k}$.
We start by initializing the message-passing with
$\msg{\delta}{\xbf_1}(\xbf_1)=\Norm(\xbf_1;\rbf_{10},\gamma_{10}^{-1}\Ibf)$.
The following steps are then repeated for $k=0,1,2,\dots$.

From Rule~\ref{rule:b}, we first set
the approximate belief on $\xbf_1$ as
$\Norm(\xbf_1;\xbfhat_{1k},\eta_{1k}^{-1}\Ibf)$,
where
$\xbfhat_{1k} = \Exp[\xbf_1|\beltil(\xbf_1)]$ and
$\eta_{1k}^{-1} = \bkt{\diag(\Cov[\xbf_1|\beltil(\xbf_1)])}$
for the SP belief
$\beltil(\xbf_1)\propto p(\xbf_1)\Norm(\xbf_1;\rbf_{1k},\gamma_{1k}^{-1}\Ibf)$.
With an i.i.d.\ prior $p(\xbf_1)$ as in \eqref{eq:pxiid},
we have that $[\xbfhat_{1k}]_n= g_1(r_{1k,n},\gamma_{1k})$
for the conditional-mean estimator $g_1(\cdot,\gamma_{1k})$
given in \eqref{eq:g1mmsesca},
yielding line~\ref{line:x1} of Algorithm~\ref{algo:vamp}.
Furthermore, from \eqref{eq:g1dervar} we see that
the corresponding conditional covariance is
$\gamma_{1k}^{-1}g_1'(r_{1k,n},\gamma_{1k})$,
yielding lines~\ref{line:a1}-\ref{line:eta1} of Algorithm~\ref{algo:vamp}.

Next, Rule~\ref{rule:v2f} says to set the message $\msg{\xbf_1}{\delta}(\xbf_1)$
proportional to
$\Norm(\xbf_1;\xbfhat_{1k},\eta_{1k}^{-1}\Ibf)
 /\Norm(\xbf_1;\rbf_{1k},\gamma_{1k}^{-1}\Ibf)$.
Since
\begin{align}
\lefteqn{
\Norm(\xbf;\xbfhat,\eta^{-1}\Ibf)/\Norm(\xbf;\rbf,\gamma^{-1}\Ibf)
}\nonumber\\
&\propto \Norm\big(\xbf;(\xbfhat\eta-\rbf\gamma)/(\eta-\gamma),(\eta-\gamma)^{-1}\Ibf\big)
\label{eq:gauss_div},
\end{align}
we have
$\msg{\xbf_1}{\delta}(\xbf_1)=\Norm(\xbf_1;\rbf_{2k},\gamma_{2k}^{-1}\Ibf)$
for
$\rbf_{2k}=(\xbfhat_{1k}\eta_{1k}-\rbf_{1k}\gamma_{1k})/(\eta_{1k}-\gamma_{1k})$
and
$\gamma_{2k}=\eta_{1k}-\gamma_{1k}$,
yielding lines~\ref{line:gam2}-\ref{line:r2} of Algorithm~\ref{algo:vamp}.
Rule~\ref{rule:f2v} then implies that the message $\msg{\xbf_1}{\delta}(\xbf_1)$
will flow rightward through the $\delta$ node unchanged, manifesting as
$\msg{\delta}{\xbf_2}(\xbf_2)=\Norm(\xbf_2;\rbf_{2k},\gamma_{2k}^{-1}\Ibf)$
on the other side.

Rule~\ref{rule:b} then says to set the approximate belief on $\xbf_2$ at
$\Norm(\xbf_2;\xbfhat_{2k},\eta_{2k}^{-1}\Ibf)$,
where
$\xbfhat_{2k} = \Exp[\xbf_2|\beltil(\xbf_2)]$ and
$\eta_{2k}^{-1} = \bkt{\diag(\Cov[\xbf_2|\beltil(\xbf_2)])}$
for the SP belief
$\beltil(\xbf_2)\propto \Norm(\xbf_2;\rbf_{2k},\gamma_{2k}^{-1}\Ibf)
 \Norm(\ybf;\Abf\xbf_2,\gamma_w^{-1}\Ibf)$.
Using standard manipulations, it can be shown that this belief is Gaussian with mean
\begin{align}
\xbfhat_{2k}
&= \left( \gamma_w \Abf\tran\Abf + \gamma_{2k}\Ibf\right)^{-1}
        \left( \gamma_w\Abf\tran\ybf + \gamma_{2k}\rbf_{2k} \right)
\label{eq:x2}
\end{align}
and covariance $(\gamma_w\Abf\tran\Abf+\gamma_{2k}\Ibf)^{-1}$.
The equivalence between \eqref{eq:x2} and \eqref{eq:g2slr}
explains line~\ref{line:x2} of Algorithm~\ref{algo:vamp}.
Furthermore, it can be seen by inspection that the average of the diagonal of this covariance matrix coincides with
$\gamma_{2k}^{-1} \bkt{\gbf_2'(\rbf_{2k},\gamma_{2k})}$
for $\bkt{\gbf_2'(\rbf_{2k},\gamma_{2k})}$ from \eqref{eq:a2slr},
thus explaining lines~\ref{line:a2}-\ref{line:eta2} of Algorithm~\ref{algo:vamp}.

Rule~\ref{rule:v2f} then says to set the message $\msg{\xbf_2}{\delta}(\xbf_2)$ at
$\Norm(\xbf_2;\xbfhat_{2k},\eta_{2k}^{-1}\Ibf)
 /\Norm(\xbf_2;\rbf_{2k},\gamma_{2k}^{-1}\Ibf)$,
which \eqref{eq:gauss_div} simplifies to
$\Norm(\xbf_2;\rbf_{1,\kp1},\gamma_{1,\kp1}^{-1}\Ibf)$
for
$\rbf_{1,\kp1}=(\xbfhat_{2k}\eta_{2k}-\rbf_{2k}\gamma_{2k})/(\eta_{2k}-\gamma_{2k})$
and
$\gamma_{1,\kp1}=\eta_{2k}-\gamma_{2k}$,
yielding lines~\ref{line:gam1}-\ref{line:r1} of Algorithm~\ref{algo:vamp}.
Finally, Rule~\ref{rule:f2v} implies that the message $\msg{\xbf_2}{\delta}(\xbf_2)$
flows left through the $\delta$ node unchanged, manifesting as
$\msg{\delta}{\xbf_1}(\xbf_1)=\Norm(\xbf_1;\rbf_{1\kp1},\gamma_{1,\kp1}^{-1}\Ibf)$
on the other side.
The above messaging sequence is then repeated with $k\leftarrow k+1$.

\section{Convergence of Vector Sequences} \label{sec:empConv}
We review some definitions from the Bayati-Montanari paper \cite{BayatiM:11}, since we will use the same
analysis framework in this paper.
Fix a dimension $r > 0$, and suppose that, for each $N$,
$\xbf(N)$ is a vector of the form
\[
    \xbf(N) = (\xbf_1(N),\ldots,\xbf_N(N)),
\]
\textb{with vector sub-components} $\xbf_n(N) \in \R^r$.  Thus, the total dimension
of $\xbf(N)$ is $rN$.  In this case, we will say that
$\xbf(N)$ is a \emph{block vector sequence that scales with $N$
under blocks $\xbf_n(N) \in \R^r$.}
When $r=1$, so that the blocks are scalar, we will simply say that
$\xbf(N)$ is a \emph{vector sequence that scales with $N$}.
Such vector sequences can be deterministic or random.
In most cases, we will omit the notational dependence on $N$ and simply write $\xbf$.

Now, given $p \geq 1$,
a function $\fbf:\R^s \arr \R^r$ is called \emph{pseudo-Lipschitz of order $p$},
if there exists a constant $C > 0$ such that for all $\xbf_1,\xbf_2 \in\R^s$,
\[
    \|\fbf(\xbf_1)-\fbf(\xbf_2)\| \leq C\|\xbf_1-\xbf_2\|\left[ 1 + \|\xbf_1\|^{p-1}
    + \|\xbf_2\|^{p-1} \right].
\]
Observe that in the case $p=1$, pseudo-Lipschitz continuity reduces to
the standard Lipschitz continuity.

\textb{Now suppose that $\xbf=\xbf(N)$ is a block vector sequence,
which may be deterministic or random.}
Given $p \geq 1$, we will say that $\xbf=\xbf(N)$ converges
\emph{empirically with $p$-th order moments} if there exists a random variable
$X \in \R^r$ such that
\begin{enumerate}[(i)]
\item $\Exp|X|^p < \infty$; and
\item for any scalar-valued pseudo-Lipschitz continuous function $f(\cdot)$ of order $p$,
\beq \label{eq:plplim}
    \lim_{N \arr \infty} \frac{1}{N} \sum_{n=1}^N f(x_n(N)) = \Exp\left[ f(X) \right] \mbox{ a.s.}.
\eeq
\end{enumerate}
Thus, the empirical mean of the components $f(x_n(N))$ converges to
the expectation $\Exp[ f(X) ]$.
\textb{When $\xbf$ converges empirically with $p$-th order moments},
we will write, with some abuse of notation,
\beq \label{eq:plLim}
    \lim_{N \arr \infty} \left\{ x_n \right\}_{n=1}^N \stackrel{PL(p)}{=} X,
\eeq
where, as usual, we have omitted the dependence $x_n=x_n(N)$.
\textb{Note that the almost sure convergence in condition (ii) applies to the case
where $\xbf(N)$ is a random vector sequence.  Importantly,
this condition holds pointwise over each function $f(\cdot)$.
It is shown in \cite[Lemma 4]{BayatiM:11} that, if condition (i) is true and
condition (ii) is true for any bounded continuous functions $f(x)$
as well as $f(x)=x^p$, then condition (ii) holds for all pseudo-Lipschitz functions
of order $p$. }

We conclude with one final definition.
Let $\phibf(\rbf,\gamma)$ be a function on $\rbf \in \R^s$ and $\gamma \in \R$.
We say that $\phibf(\rbf,\gamma)$ is \emph{uniformly Lipschitz continuous} in $\rbf$
at $\gamma=\gammabar$ if there exists constants
$L_1$ and $L_2 \geq 0$ and an open neighborhood $U$ of $\gammabar$, such that
\beq \label{eq:unifLip1}
    \|\phibf(\rbf_1,\gamma)-\phibf(\rbf_2,\gamma)\| \leq L_1\|\rbf_1-\rbf_2\|,
\eeq
for all $\rbf_1,\rbf_2 \in \R^s$ and $\gamma \in U$; and
\beq \label{eq:unifLip2}
    \|\phibf(\rbf,\gamma_1)-\phibf(\rbf,\gamma_2)\| \leq L_2\left(1+\|\rbf\|\right)|\gamma_1-\gamma_2|,
\eeq
for all $\rbf \in \R^s$ and $\gamma_1,\gamma_2 \in U$.

\section{Proof of Lemmas~\ref{lem:errfn} and \ref{lem:sens}} \label{sec:errsenspf}

For Lemma~\ref{lem:errfn}, part (a) follows immediately from
\eqref{eq:g1mmsex0} and \eqref{eq:eps1}.
To prove part (b), suppose
\[
    \ybf = \Abf\xbf^0 + \wbf, \quad \rbf_2 = \xbf^0+\qbf.
\]
Then, the error is given by
\begin{align}
    \MoveEqLeft \gbf_2(\rbf_{2},\gamma_2) - \xbf^0
    \stackrel{(a)}{=} \left( \gamma_w \Abf\tran\Abf + \gamma_2\Ibf\right)^{-1}
        \nonumber \\
    & \quad \times \left( \gamma_w\Abf\tran\Abf\xbf^0 + \gamma_w\Abf\tran\wbf
    + \gamma_2\rbf_{2} \right) -\xbf^0 \nonumber \\
   &\stackrel{(b)}{=}
   \left( \gamma_w \Abf\tran\Abf + \gamma_2\Ibf\right)^{-1}
    \left( \gamma_2\qbf+ \gamma_w\Abf\tran\wbf\right),  \nonumber \\
   &\stackrel{(c)}{=} \Qbf^{-1}
    \left( \gamma_2\qbf+ \gamma_w\Abf\tran\wbf\right), \nonumber
\end{align}
where (a) follows by substituting $\ybf = \Abf\xbf^0 + \wbf$ into \eqref{eq:g2slr};
part (b) follows from the substitution $\textb{\rbf_2} = \xbf^0+\qbf$ and collecting the terms
with $\xbf^0$; and (c) follows from the definition of $\Qbf$ in \eqref{eq:QQtdef}.
Hence, the error covariance matrix is given
\begin{align*}
    \MoveEqLeft \Exp\left[ (\gbf_2(\rbf_{2},\gamma_2) - \xbf^0)(\gbf_2(\rbf_{2},\gamma_2) - \xbf^0)\tran
        \right] \\
    &= \Qbf^{-1}\left[ \gamma_2^2\Exp[\qbf\qbf\tran] + \gamma_w^2\Abf\Exp[\wbf\wbf\tran]\Abf\tran
    \right]\Qbf^{-1} \\
    &= \Qbf^{-1}\tilde{\Qbf}\Qbf^{-1},
\end{align*}
where we have used the
the fact that $\qbf$ and $\wbf$ are independent Gaussians
with variances $\tau_2$ and $\gamma_{w0}^{-1}$.    This proves \eqref{eq:eps2Q}.
\textb{Then,} under the matched condition, \textb{we have that} $\Qbf=\tilde{\Qbf}$, which proves \eqref{eq:eps2Qmatch}.
Part (c) of Lemma~\ref{lem:errfn} follows from part (b) by using the SVD \eqref{eq:ASVD}.

For Lemma~\ref{lem:sens}, part (a) follows from averaging \eqref{eq:g1mmsex0} over $r_1$.
Part (b) follows by taking the derivative in \eqref{eq:g2slr} and part (c)
follows from using the SVD~\eqref{eq:ASVD}.

\section{Orthogonal Matrices Under Linear Constraints}

In preparation for proving Theorem~\ref{thm:se},
we derive various results on orthogonal matrices subject to linear constraints.
To this end, suppose $\Vbf \in \R^{N \x N}$ is an orthogonal matrix
satisfying linear constraints
\beq \label{eq:AVB}
    \Abf = \Vbf\Bbf,
\eeq
for some matrices $\Abf, \Bbf \in \R^{N \x s}$ for some $s$.  Assume $\Abf$ and $\Bbf$
are full column rank (hence $s \leq N$).  Let
\beq \label{eq:UABdef}
    \Ubf_\Abf = \Abf(\Abf\tran\Abf)^{-1/2}, \quad
    \Ubf_\Bbf = \Bbf(\Bbf\tran\Bbf)^{-1/2}.
\eeq
Also, let $\Ubf_{\Abf^\perp}$ and $\Ubf_{\Bbf^\perp}$ be any $N \x (N-s)$
matrices whose columns are
an orthonormal bases for $\mathrm{Range}(\Abf)^\perp$ and
$\mathrm{Range}(\Bbf)^\perp$, respectively.
Define
\beq \label{eq:Vtdef}
    \tilde{\Vbf} := \Ubf_{\Abf^\perp}\tran\Vbf\Ubf_{\Bbf^\perp},
\eeq
which has dimension $(N-s) \x (N-s)$.

\begin{lemma}  \label{lem:orthogRep}
Under the above definitions $\tilde{\Vbf}$ satisfies
\beq \label{eq:VVtrep}
    \Vbf = \Abf(\Abf\tran\Abf)^{-1}\Bbf\tran + \Ubf_{\Abf^\perp}\tilde{\Vbf}\Ubf_{\Bbf^\perp}\tran.
\eeq
\end{lemma}
\begin{proof}
\textb{Let $\Pbf_\Abf:=\Ubf_\Abf\Ubf_\Abf\tran$ and $\Pbf_\Abf^\perp:=\Ubf_{\Abf^\perp}\Ubf_{\Abf^\perp}\tran$ are
the orthogonal projections onto $\mathrm{Range}(\Abf)$ and $\mathrm{Range}(\Abf)^\perp$ respectively.
Define $\Pbf_\Bbf$ and $\Pbf_\Bbf^\perp$ similarly.  Since, $\Abf=\Vbf\Bbf$, we have $\Vbf\tran\Abf = \Bbf$ and therefore,
\beq \label{eq:PAVPB}
    \Pbf_\Abf^\perp\Vbf\Pbf_{\Bbf} = \mathbf{0}, \quad
    \Pbf_\Abf\Vbf\Pbf_{\Bbf}^\perp = \mathbf{0}.
\eeq
Therefore,
\begin{align}
    \Vbf &= (\Pbf_\Abf+\Pbf_\Abf^\perp)\Vbf(\Pbf_\Bbf+\Pbf_\Bbf^\perp) \nonumber \\
    &= (\Pbf_\Abf\Vbf\Pbf_\Bbf + \Pbf_\Abf^\perp\Vbf\Pbf_\Bbf^\perp). \label{eq:PVP1}
\end{align}
Now,
\begin{align}
    \MoveEqLeft \Pbf_\Abf\Vbf\Pbf_\Bbf = \Pbf_\Abf\Vbf\Bbf(\Bbf\Bbf\tran)^{-1}\Bbf\tran \nonumber \\
    &= \Pbf_{\Abf}\Abf(\Bbf\Bbf\tran)^{-1}\Bbf\tran \nonumber \\
    &= \Abf(\Bbf\Bbf\tran)^{-1}\Bbf\tran = \Abf(\Abf\tran\Abf)^{-1}\Bbf\tran, \label{eq:PVP2}
\end{align}
where, in the last step we used the fact that
\[
    \Abf\tran\Abf=\Bbf\tran\Vbf\tran\Vbf\Bbf = \Bbf\tran\Bbf.
\]
Also,  using the definition of $\tilde{\Vbf}$ in \eqref{eq:Vtdef},
\beq \label{eq:PVP3}
    \Pbf_{\Abf^\perp}\Vbf\Pbf_{\Bbf^\perp}\tran  =\Ubf_{\Abf^\perp}\tilde{\Vbf}\Ubf_{\Bbf^\perp}\tran.
\eeq
Substituting \eqref{eq:PVP2} and \eqref{eq:PVP3} into \eqref{eq:PVP1} obtains \eqref{eq:VVtrep}.
To prove that $\tilde{\Vbf}$ is orthogonal,
\begin{align*}
    \MoveEqLeft \tilde{\Vbf}\tran\tilde{\Vbf} \stackrel{(a)}{=}
    \Ubf_{\Bbf^\perp}\tran\Vbf\Pbf_\Abf\Vbf\Ubf_{\Bbf^\perp}
    \\
    &\stackrel{(b)}{=}
    \Ubf_{\Bbf^\perp}\tran\Vbf\tran\Vbf\Ubf_{\Bbf^\perp} \stackrel{(c)}{=} \Ibf,
\end{align*}
where (a) uses \eqref{eq:Vtdef}; (b) follows from \eqref{eq:PAVPB} and
(c) follows from the fact that $\Vbf$ and $\Ubf_{\Bbf^\perp}$ have orthonormal
columns.
}
\end{proof}

\begin{lemma} \label{lem:orthogLin}  Let $\Vbf \in \R^{N \x N}$ be a random matrix
that is Haar distributed.  Suppose that $\Abf$ and $\Bbf$ are deterministic and $G$ is
the event that $\Vbf$ satisfies linear constraints \eqref{eq:AVB}.
Then, the conditional distribution given $G$,
$\tilde{\Vbf}$ is Haar distributed matrix independent of $G$.  Thus,
\[
    \left. \Vbf \right|_{G} \eqd
    \Abf(\Abf\tran\Abf)^{-1}\Bbf\tran + \Ubf_{\Abf^\perp}\tilde{\Vbf}\Ubf_{\Bbf^\perp}\tran,
\]
where $\tilde{\Vbf}$ is Haar distributed and independent of $G$.
\end{lemma}
\begin{proof}   Let $O_N$  be the set
of $N \x N$ orthogonal matrices and let ${\mathcal L}$ be the set of matrices $\Vbf \in O_N$
that satisfy the linear constraints \eqref{eq:AVB}.  If $p_{\Vbf}(\Vbf)$ is the uniform density
on $O_N$ (i.e.\ the Haar measure), the conditional density on $\Vbf$ given the event $G$,
\[
    p_{\Vbf|G}(\Vbf|G) = \frac{1}{Z}p_{\Vbf}(\Vbf)\indic{\Vbf \in {\mathcal L}},
\]
where $Z$ is the normalization constant.  Now let $\phi:\tilde{\Vbf} \mapsto \Vbf$
be the mapping described by \eqref{eq:VVtrep} which maps $O_{N-s}$ to ${\mathcal L}$.
This mapping is invertible.  Since $\phi$ is affine, the conditional density on
$\tilde{\Vbf}$ is given by
\begin{align}
    \MoveEqLeft p_{\tilde{\Vbf}|G}(\tilde{\Vbf}|G) \propto
    p_{\Vbf|G}(\phi(\tilde{\Vbf})|G) \nonumber \\
    &\propto p_{\Vbf}(\phi(\tilde{\Vbf}))\indic{\phi(\tilde{\Vbf} \in {\mathcal L})}
    = p_{\Vbf}(\phi(\tilde{\Vbf})), \label{eq:ptvg}
\end{align}
where in the last step we used the fact that, for any matrix $\tilde{\Vbf}$,
$\phi(\tilde{\Vbf}) \in {\mathcal L}$ (i.e.\ satisfies the linear constraints \eqref{eq:AVB}).
Now to show that $\tilde{\Vbf}$ is conditionally Haar distributed, we need to show that for any
orthogonal matrix $\Wbf_0 \in O_{N-s}$,
\beq \label{eq:pvtorthog}
    p_{\tilde{\Vbf}|G}(\Wbf_0\tilde{\Vbf}|G)=p_{\tilde{\Vbf}|G}(\tilde{\Vbf}|G).
\eeq
To prove this, given $\Wbf_0 \in O_{N-s}$, define the matrix,
\[
    \Wbf = \Ubf_{\Abf}\Ubf_{\Abf}\tran + \Ubf_{\Abf^\perp}\Wbf_0\Ubf_{\Abf^\perp}\tran.
\]
One can verify that $\Wbf \in O_N$ (i.e.\ it is orthogonal) and
\beq \label{eq:pwvt}
    \phi(\Wbf_0\tilde{\Vbf}) =\Wbf\phi(\tilde{\Vbf}).
\eeq
Hence,
\begin{align}
    \MoveEqLeft p_{\tilde{\Vbf}|G}(\Wbf_0\tilde{\Vbf}|G)
    \stackrel{(a)}{\propto} p_{\Vbf}(\phi(\Wbf_0\tilde{\Vbf})) \nonumber \\
    &\stackrel{(b)}{\propto} p_{\Vbf}(\Wbf\phi(\tilde{\Vbf}))
    \stackrel{(c)}{\propto} p_{\Vbf}(\phi(\tilde{\Vbf})), \nonumber
\end{align}
where (a) follows from \eqref{eq:ptvg}; (b) follows from \eqref{eq:pwvt}; and
(c) follows from the orthogonal invariance of $\Vbf$.  Hence, the conditional density
of $\tilde{\Vbf}$ is invariant under orthogonal transforms and is thus Haar distributed.
\end{proof}

We will use Lemma~\ref{lem:orthogLin} in conjunction with the following
simple result.

\begin{lemma} \label{lem:orthogGaussLim} Fix a dimension $s \geq 0$,
and suppose that $\xbf(N)$ and $\Ubf(N)$ are sequences such that
for each $N$,
\begin{enumerate}[(i)]
\item $\Ubf=\Ubf(N) \in \R^{N\x (N-s)}$ is a deterministic
matrix with $\Ubf\tran\Ubf = \Ibf$;
\item $\xbf=\xbf(N) \in \R^{N-s}$ a random vector that
is isotropically distributed  in that $\Vbf\xbf\eqd\xbf$ for any orthogonal $(N-s)\x(N-s)$ matrix
$\Vbf$.
\item The \textb{normalized squared Euclidean norm} converges almost surely as
\[
    \lim_{N \arr \infty} \frac{1}{N} \|\xbf\|^2 = \tau,
\]
for some $\tau > 0$.
\end{enumerate}
Then, if we define $\ybf = \Ubf\xbf$, we have that the components of $\ybf$
converge empirically to a Gaussian random variable
\beq \label{eq:ygausslim}
    \lim_{N \arr \infty} \{ y_n \} \PLeq Y \sim \Norm(0,\tau).
\eeq
\end{lemma}
\begin{proof}  Since $\xbf$ is isotropically distributed, it can be generated as a normalized
Gaussian, i.e.\
\[
    \xbf \eqd \frac{\|\xbf\|}{\|\wbf_0\|}\wbf_0, \quad \wbf_0 \sim \Norm(\mathbf{0},\Ibf_{N-s}).
\]
For each $N$, let $\Ubf_\perp$ be an $N \x s$ matrix such that $\Sbf := [\Ubf ~ \Ubf_{\perp}]$ is
orthogonal.  That is, the $s$ columns of $\Ubf_\perp$ are an orthonormal basis
of the orthogonal complement of the $\Range(\Ubf)$.  If we let $\wbf_1 \sim \Norm(0,\Ibf_s)$,
then if we define
\[
    \wbf = \left[ \begin{array}{c} \wbf_0 \\ \wbf_1 \end{array} \right],
\]
so that $\wbf \sim \Norm(0,\Ibf_N)$.
With this definition, we can write $\ybf$ as
\beq \label{eq:yux}
  \ybf = \Ubf\xbf \eqd \frac{\|\xbf\|}{\|\wbf_0\|}\left[ \Sbf\wbf - \Ubf_{\perp}\wbf_1 \right].
\eeq
Now,
\[
    \lim_{N \arr\infty} \frac{\|\xbf\|}{\|\wbf_0\|} = \sqrt{\tau},
\]
almost surely.  Also, since $\wbf \sim \Norm(0,\Ibf)$ and $\Sbf$ is orthogonal,
$\Sbf\wbf \sim \Norm(0,\Ibf)$.  Finally, since $\wbf_1$ is $s$-dimensional,
\[
    \lim_{N \arr \infty} \frac{1}{N} \|\Ubf_{\perp}\wbf_1\|^2 =
    \lim_{N \arr \infty} \frac{1}{N} \|\wbf_1\|^2 = 0,
\]
almost surely.  Substituting these properties into \eqref{eq:yux},
we obtain \eqref{eq:ygausslim}.
\end{proof}

\section{A General Convergence Result}

To analyze the VAMP method, we a consider the
following more general recursion.  For each dimension $N$, we are given
an orthogonal matrix $\Vbf \in \R^{N \x N}$,
and an initial vector $\ubf_0 \in \R^N$.  
\textb{Also, we are given disturbance vectors 
\[
    \wbf^p=(w_1^p,\ldots,w_n^p), \quad
    \wbf^q=(w_1^q,\ldots,w_n^q),
\]
where the components $w_n^p \in \R^{n_p}$ and $w_n^q \in \R^{n_q}$ 
for some finite dimensions $n_p$ and $n_q$ that do not grow with $N$.
}
Then, we generate a sequence of iterates by the following recursion:%
\begin{subequations}
\label{eq:algoGen}
\begin{align}
    \pbf_k &= \Vbf\ubf_k \label{eq:pupgen} \\
    \alpha_{1k} &= \bkt{ \fbf_p'(\pbf_k,\wbf^p,\gamma_{1k})},
    \quad \gamma_{2k} = \Gamma_1(\gamma_{1k},\alpha_{1k})
        \label{eq:alpha1gen}  \\
    \vbf_k &= C_1(\alpha_{1k})\left[
        \fbf_p(\pbf_k,\wbf^p,\gamma_{1k})- \alpha_{1k} \pbf_{k} \right] \label{eq:vupgen} \\
    \qbf_k &= \Vbf\tran\vbf_k \label{eq:qupgen} \\
    \alpha_{2k} &= \bkt{ \fbf_q'(\qbf_k,\wbf^q,\gamma_{2k})},
    \quad \gamma_{1,\kp1} =\Gamma_2(\gamma_{2k},\alpha_{2k})
        \label{eq:alpha2gen} \\
    \ubf_{\kp1} &= C_2(\alpha_{2k})\left[
        \fbf_q(\qbf_k,\wbf^q,\gamma_{2k}) - \alpha_{2k}\qbf_{k} \right], \label{eq:uupgen}
\end{align}
\end{subequations}
which is initialized with some vector $\ubf_0$ and scalar $\gamma_{10}$.
Here, $\fbf_p(\cdot)$ and $\fbf_q(\cdot)$ are separable functions, meaning
\begin{align}\label{eq:fpqcomp}
\begin{split}
    \left[ \fbf_p(\pbf,\wbf^p,\gamma_1)\right]_n = f_p(p_n,w^p_n,\gamma_1)~\forall n, \\
    \left[ \fbf_q(\qbf,\wbf^q,\gamma_2)\right]_n = f_q(q_n,w^q_n,\gamma_2)~\forall n,
\end{split}
\end{align}
for scalar-valued functions $f_p(\cdot)$ and $f_q(\cdot)$.
The functions $\Gamma_i(\cdot)$ and $C_i(\cdot)$ are also scalar-valued.
\textb{In the recursion \eqref{eq:algoGen}, the variables $\gamma_{1k}$
and $\gamma_{2k}$ 
represent some parameter of the update functions $\fbf_p(\cdot)$ and $\fbf_q(\cdot)$,
and the functions $\Gamma_i(\cdot)$ represent how these parameters are updated.}

Similar to our analysis of the VAMP, we consider the following
large-system limit (LSL) analysis.
We consider a sequence of runs of the recursions indexed by $N$.
We model the initial condition $\ubf_0$ and disturbance vectors $\wbf^p$ and $\wbf^q$
as deterministic sequences that scale with $N$ and assume that their components
converge empirically as
\beq \label{eq:U0lim}
    \lim_{N \arr \infty} \{ u_{0n} \} \PLeq U_0,
\eeq
and
\beq \label{eq:WpqLim}
    \lim_{N \arr \infty} \{ w^p_n \} \PLeq W^p, \quad
    \lim_{N \arr \infty} \{ w^q_n \} \PLeq W^q,
\eeq
to random variables $U_0$, $W^p$ and $W^q$.  \textb{The
vectors $W_p$ and $W_q$ are random vectors in $\R^{n_p}$ and $\R^{n_q}$,
respectively.} We assume that the
initial constant converges as
\beq \label{eq:gam10limgen}
    \lim_{N \arr \infty} \gamma_{10} = \gammabar_{10},
\eeq
for some $\gammabar_{10}$.
  The matrix $\Vbf \in \R^{N \x N}$
is assumed to be uniformly distributed on the set of orthogonal matrices
independent of $\rbf_0$, $\wbf^p$ and $\wbf^q$.
Since $\rbf_0$, $\wbf^p$ and
$\wbf^q$ are deterministic, the only randomness is in the matrix $\Vbf$.

Under the above assumptions, define the SE equations
\begin{subequations} \label{eq:segen}
\begin{align}
    \alphabar_{1k} &= \Exp\left[ f_p'(P_k,W^p,\gammabar_{1k})\right],
        \label{eq:a1segen} \\
    \tau_{2k} &= C_1^2(\alphabar_{1k}) \left\{
        \Exp\left[ f_p^2(P_k,W^p,\gammabar_{1k})\right]  - \alphabar_{1k}^2\tau_{1k} \right\}
        \label{eq:tau2segen} \\
    \gammabar_{2k} &= \Gamma_1(\gammabar_{1k},\alphabar_{1k}) \label{eq:gam2segen} \\
    \alphabar_{2k} &= \Exp\left[ f_q'(Q_k,W^q,\gammabar_{2k} \right],
        \label{eq:a2segen} \\
    \tau_{1,\kp1} &= C_2^2(\alphabar_{2k})\left\{
        \Exp\left[ f_q^2(Q_k,W^q,\gammabar_{2k})\right] - \alphabar_{2k}^2\tau_{2k}\right\}
        \label{eq:tau1segen} \\
    \gamma_{1,\kp1} &= \Gamma_2(\gammabar_{2k},\alphabar_{2k}), \label{eq:gam1segen}
\end{align}
\end{subequations}
which are initialized with $\gammabar_{10}$ in \eqref{eq:gam10limgen} and
\beq \label{eq:tau10gen}
    \tau_{10} = \Exp[U_0^2],
\eeq
where $U_0$ is the random variable in \eqref{eq:U0lim}.
In the SE equations~\eqref{eq:segen},
the expectations are taken with respect to random variables
\[
    P_k \sim \Norm(0,\tau_{1k}), \quad Q_k \sim \Norm(0,\tau_{2k}),
\]
where $P_k$ is independent of $W^p$ and $Q_k$ is independent of $W^q$.

\begin{theorem} \label{thm:genConv}  Consider the recursions \eqref{eq:algoGen}
and SE equations \eqref{eq:segen} under the above assumptions.  Assume additionally that,
for all $k$:
\begin{enumerate}[(i)]
\item For $i=1,2$, the functions
\[
    C_i(\alpha_i), \quad \Gamma_i(\gamma_i,\alpha_i),
\]
are continuous at the points $(\gamma_i,\alpha_i)=(\gammabar_{ik},\alphabar_{ik})$
from the SE equations; and
\item The function $f_p(p,w^p,\gamma_1)$ and its derivative $f_p'(p,w^p,\gamma_1)$
are uniformly Lipschitz continuous in $(p,w^p)$ at $\gamma_1=\gammabar_{1k}$.
\item The function $f_q(q,w^q,\gamma_2)$ and its derivative $f_q'(q,w^q,\gamma_2)$
are uniformly Lipschitz continuous in $(q,w^q)$ at $\gamma_2=\gammabar_{2k}$.
\end{enumerate}
Then,
\begin{enumerate}[(a)]
\item For any fixed $k$, almost surely the components of $(\wbf^p,\pbf_0,\ldots,\pbf_k)$
empirically converge as
\beq \label{eq:Pconk}
    \lim_{N \arr \infty} \left\{ (w^p_n,p_{0n},\ldots,p_{kn}) \right\}
    \PLeq (W^p,P_0,\ldots,P_k),
\eeq
where $W^p$ is the random variable in the limit \eqref{eq:WpqLim} and
$(P_0,\ldots,P_k)$ is a zero mean Gaussian random vector independent of $W^p$,
with $\Exp[P_k^2] = \tau_{1k}$.  In addition, we have that
\beq \label{eq:ag1limgen}
    \lim_{N \arr \infty} (\alpha_{1k},\gamma_{1k}) = (\alphabar_{1k},\gammabar_{1k}),
\eeq
almost surely.

\item For any fixed $k$, almost surely the components of $(\wbf^q,\qbf_0,\ldots,\qbf_k)$
empirically converge as
\beq \label{eq:Qconk}
    \lim_{N \arr \infty} \left\{ (w^q_n,q_{0n},\ldots,q_{kn}) \right\}
    \PLeq (W^q,Q_0,\ldots,Q_k),
\eeq
where $W^q$ is the random variable in the limit \eqref{eq:WpqLim} and
$(Q_0,\ldots,Q_k)$ is a zero mean Gaussian random vector independent of $W^q$,
with $\Exp[P_k^2] = \tau_{2k}$.  In addition, we have that
\beq \label{eq:ag2limgen}
    \lim_{N \arr \infty} (\alpha_{2k},\gamma_{2k}) = (\alphabar_{2k},\gammabar_{2k}),
\eeq
almost surely.
\end{enumerate}
\end{theorem}
\begin{proof}  We will prove this in the next Appendix,
Appendix~\ref{sec:genConvPf}.
\end{proof}

\section{Proof of Theorem~\ref{thm:genConv}} \label{sec:genConvPf}

\subsection{Induction Argument}
We use an induction argument.   Given iterations $k, \ell \geq 0$,
define the hypothesis, $H_{k,\ell}$ as the statement:
\begin{itemize}
\item Part (a) of Theorem~\ref{thm:genConv} is true up to $k$; and
\item Part (b) of Theorem~\ref{thm:genConv} is true up to $\ell$.
\end{itemize}
The induction argument will then follow by showing the following three facts:
\begin{itemize}
\item $H_{0,-1}$ is true;
\item If $H_{k,\km1}$ is true, then so is $H_{k,k}$;
\item If $H_{k,k}$ is true, then so is $H_{\kp1,k}$.
\end{itemize}

\subsection{Induction Initialization}
We first show that the hypothesis $H_{0,-1}$ is true.  That is,
we must show \eqref{eq:Pconk} and \eqref{eq:ag1limgen} for $k=0$.
This is a special case of Lemma~\ref{lem:orthogGaussLim}.
Specifically, for each $N$, let $\Ubf=\Ibf_N$, the $N \x N$ identity
matrix, which trivially satisfies property (i) of Lemma~~\ref{lem:orthogGaussLim}
with $s=0$.  Let $\xbf = \pbf_0$.  Since $\pbf_0 = \Vbf\ubf_0$ and $\Vbf$ is
Haar distributed independent of $\ubf_0$, we have that $\pbf_0$ is orthogonally
invariant and satisfies property (ii) of Lemma~\ref{lem:orthogGaussLim}.
Also,
\[
    \lim_{N \arr \infty} \|\pbf_0\|^2 \stackrel{(a)}{=}
     \lim_{N \arr \infty} \|\ubf_0\|^2 \stackrel{(b)}{=} \Exp [U_0^2]
     \stackrel{(c)}{=} \tau_{10},
\]
where (a) follows from the fact that  $\pbf_0=\Vbf\ubf_0$ and $\Vbf$ is orthogonal;
(b) follows from the assumption~\eqref{eq:U0lim} and (c) follows from the definition
\eqref{eq:tau10gen}.  This proves property (iii) of Lemma~\ref{lem:orthogGaussLim}.
Hence, $\pbf_0 = \Ubf\pbf_0$,  we have that the components of $\pbf_0$ converge
empirically as
\[
    \lim_{N \arr \infty} \{ p_{0n} \} \PLeq P_0 \sim \Norm(0,\tau_{10}),
\]
for a Gaussian random variable $P_0$.  Moreover, since $\Vbf$ is independent of $\wbf^p$,
and the components of $\wbf^p$ converge empirically as \eqref{eq:WpqLim},
we have that the components of $\pbf_n,\wbf^p$ almost surely converge empirically as
\[
    \lim_{N \arr \infty} \{ w^p_n, p_{0n} \} \PLeq (W^p,P_0),
\]
where $W^p$ is independent of $P_0$.  This  proves \eqref{eq:Pconk} for $k=0$.

Now, we have assumed in \eqref{eq:gam10limgen} that $\gamma_{10} \arr \gammabar_{10}$
as $N \arr \infty$.  Also, since
 $f_p'(p,w^p,\gamma_1)$ is uniformly Lipschitz continuous in $(p,w^p)$
at $\gamma_1 = \gammabar_{10}$, we have that
$\alpha_{10} = \bkt{ \fbf_p'(\pbf_0,\wbf^p,\gamma_{10}) }$
converges to $\alphabar_{10}$ in \eqref{eq:a1segen} almost surely.
This proves \eqref{eq:ag1limgen}.

\subsection{The Induction Recursion}
We next show the implication $H_{k,\km1} \Rightarrow H_{k,k}$. The implication
$H_{k,k} \Rightarrow H_{\kp1,k}$ is proven similarly.  Hence, fix $k$ and assume
that $H_{k,\km1}$ holds.
Since $\Gamma_1(\gamma_i,\alpha_i)$ is continuous at $(\gammabar_{1k},\alphabar_{1k})$,
the limits \eqref{eq:ag1limgen} combined with \eqref{eq:gam2segen} show that
\[
    \lim_{N \arr\infty} \gamma_{2k} =\lim_{N \arr\infty} \Gamma_1(\gamma_{1k},\alpha_{1k}) =
        \gammabar_{2k}.
\]
In addition, the induction hypothesis shows that for $\ell=0,\ldots,k$,
the components of $(\wbf^p,\pbf_\ell)$ almost surely converge empirically as
\[
    \lim_{N \arr\infty} \{ (w^p_n,p_{\ell n})\}  \PLeq (W^p,P_\ell),
\]
where $P_\ell \sim \Norm(0,\tau_{1\ell})$ for $\tau_{1\ell}$ given by the SE equations.
Since $f_p(\cdot)$ is Lipschitz continuous \textb{and} $C_1(\alpha_{1\ell})$ is continuous
at $\alpha_{1\ell}=\alphabar_{1\ell}$, \textb{one may observe that} the definition of $\vbf_\ell$ in \eqref{eq:vupgen}
and the limits \eqref{eq:ag1limgen} show that
\[
    \lim_{N \arr\infty} \{ (w^p_n,p_{\ell n},v_{\ell n})\}  \PLeq (W^p,P_\ell ,V_\ell ),
\]
where $V_\ell $ is the random variable
\beq \label{eq:Velldef}
    V_\ell  = g_p(P_\ell ,W^p,\gammabar_{1\ell },\alphabar_{1\ell}),
\eeq
and $g_p(\cdot)$ is the function
\beq \label{eq:gpdef}
    g_p(p,w^p,\gamma_1,\alpha_1) := C_1(\alpha_1)\left[ f_p(p,w^p,\gamma_1)-\alpha_1 p \right].
\eeq
\textb{
Similarly, we have the limit
\[
    \lim_{N \arr\infty} \{ (w^q_n,q_{\ell n},u_{\ell n})\}  
    \PLeq (W^q,Q_\ell ,U_\ell ),
\]
where $U_\ell $ is the random variable,
\beq \label{eq:Uelldef}
    U_\ell  = g_q(Q_\ell ,W^q,\gammabar_{2\ell },\alphabar_{2\ell})
\eeq
and $g_q(\cdot)$ is the function
\beq \label{eq:gqdef}
    g_q(q,w^q,\gamma_2,\alpha_2) := C_2(\alpha_1)\left[ f_q(q,w^q,\gamma_2)-\alpha_2 q \right].
\eeq
}

We next introduce the notation
\[
    \Ubf_k := \left[ \ubf_0 \cdots \ubf_k \right] \in \R^{N \x (\kp1)},
\]
to represent the first $\kp1$ values of the vector $\ubf_\ell$.
We define the matrices $\Vbf_k$, $\Qbf_k$ and $\Pbf_k$ similarly.
Using this notation, let $G_k$ be the \textb{tuple of random matrices},
\beq \label{eq:Gdef}
    G_k := \left\{ \Ubf_k, \Pbf_k, \Vbf_k, \Qbf_{\km1} \right\}.
\eeq
With some abuse of notation, we will also use $G_k$
to denote the sigma-algebra generated by these variables.
The set \eqref{eq:Gdef} contains all the outputs of the algorithm
\eqref{eq:algoGen} immediately \emph{before} \eqref{eq:qupgen} in iteration $k$.

Now, the actions of the matrix $\Vbf$ in the recursions \eqref{eq:algoGen}
are through the matrix-vector multiplications \eqref{eq:pupgen} and \eqref{eq:qupgen}.
Hence, if we define the matrices
\beq \label{eq:ABdef}
    \Abf_k := \left[ \Pbf_k ~ \Vbf_{\km1} \right], \quad
    \Bbf_k := \left[ \Ubf_k ~ \Qbf_{\km1} \right],
\eeq
the output of the recursions in the set $G_k$ will be unchanged for all
matrices $\Vbf$ satisfying the linear constraints
\beq \label{eq:ABVconk}
    \Abf_k = \Vbf\Bbf_k.
\eeq
Hence, the conditional distribution of $\Vbf$ given $G_k$ is precisely
the uniform distribution on the set of orthogonal matrices satisfying
\eqref{eq:ABVconk}.  The matrices $\Abf_k$ and $\Bbf_k$ are of dimensions
$N \x s$ where $s=2k+1$.
From Lemma~\ref{lem:orthogLin}, this conditional distribution is given by
\beq \label{eq:Vconk}
    \left. \Vbf \right|_{G_k} \eqd
    \Abf_k(\Abf\tran_k\Abf_k)^{-1}\Bbf_k\tran + \Ubf_{\Abf_k^\perp}\tilde{\Vbf}\Ubf_{\Bbf_k^\perp}\tran,
\eeq
where $\Ubf_{\Abf_k^\perp}$ and $\Ubf_{\Bbf_k^\perp}$ are $N \x (N-s)$ matrices
whose columns are an orthonormal basis for $\Range(\Abf_k)^\perp$ and $\Range(\Bbf_k)^\perp$.
The matrix $\tilde{\Vbf}$ is  Haar distributed on the set of $(N-s)\x(N-s)$
orthogonal matrices and independent of $G_k$.

Using \eqref{eq:Vconk} we can write $\qbf_k$ in \eqref{eq:qupgen} as a sum of two terms
\beq \label{eq:qpart}
    \qbf_k = \Vbf\tran\vbf_k = \qbf_k^{\rm det} + \qbf_k^{\rm ran},
\eeq
where $\qbf_k^{\rm det}$ is what we will call the \emph{deterministic} part:
\beq \label{eq:qkdet}
    \qbf_k^{\rm det} = \Bbf_k(\Abf\tran_k\Abf_k)^{-1}\Abf_k\tran\vbf_k,
\eeq
and $\qbf_k^{\rm ran}$ is what we will call the \emph{random} part:
\beq \label{eq:qkran}
    \qbf_k^{\rm ran} = \Ubf_{\Bbf_k^\perp}\tilde{\Vbf}\tran \Ubf_{\Abf_k^\perp}\tran \vbf_k.
\eeq
The next few lemmas will evaluate the asymptotic distributions of the two
terms in \eqref{eq:qpart}.

\begin{lemma} \label{lem:qconvdet}
Under the induction hypothesis $H_{k,\km1}$, there \textb{exist} constants
$\beta_{k,0},\ldots,\beta_{k,\km1}$ such that the components of $\qbf_k^{\rm det}$
along with $(\qbf_0,\ldots,\qbf_{\km1})$ converge empirically as
\begin{align}
    \MoveEqLeft \lim_{N \arr \infty} \left\{ w^q_n,q_{0n},\ldots,q_{\km1,n},q_{kn}^{\rm det}) \right\}
    \nonumber \\
    &\PLeq (W^q,Q_0,\ldots,Q_{\km1},Q_k^{\rm det}),
        \label{eq:qconvdet}
\end{align}
where $Q_\ell$, $\ell=0,\ldots,\km1$ are the Gaussian random variables in induction hypothesis
\eqref{eq:Qconk} and $Q_k^{\rm det}$ is a linear combination,
\beq \label{eq:Qkdetlim}
    Q_k^{\rm det} = \beta_{k0}Q_0 + \cdots + \beta_{k,\km1}Q_{\km1}.
\eeq
\end{lemma}
\begin{proof}
We evaluate the asymptotic values of various terms in \eqref{eq:qkdet}.
Using the definition of $\Abf_k$ in \eqref{eq:ABdef},
\[
    \Abf\tran_k\Abf_k = \left[ \begin{array}{cc}
        \Pbf_k\tran\Pbf_k & \Pbf_k\tran\Vbf_{\km1} \\
        \Vbf_{\km1}\tran\Pbf_k & \Vbf_{\km1}\tran\Vbf_{\km1}
        \end{array} \right]
\]
We can then easily evaluate the asymptotic value of these
terms as follows.  For example, the asymptotic value of the
$(i,j)$ component of the matrix $\Pbf_k\tran\Pbf_k$ is given by
\begin{align*}
    \MoveEqLeft \lim_{N \arr \infty} \frac{1}{N} \left[ \Pbf_k\tran\Pbf_k \right]_{ij}
        \stackrel{(a)}{=} \frac{1}{N} \pbf_i\tran\pbf_j \\
        &= \frac{1}{N} \sum_{n=1}^N p_{in}p_{jn}
        \stackrel{(b)}{=} E(P_iP_j) \stackrel{(c)}{=} \left[ \Qbf^p_k\right]_{ij},
\end{align*}
where (a) follows since the $i$-th column of $\Pbf_k$ is precisely the vector
$\pbf_i$; (b) follows due to convergence assumption in \eqref{eq:Pconk};
and in (c), we use $\Qbf^p_k$ to denote the covariance matrix of $(P_0,\ldots,P_k)$.
Similarly
\[
    \lim_{N \arr \infty}  \frac{1}{N}  \Vbf_{\km1}\tran\Vbf_{\km1} = \Qbf^v_{\km1},
\]
where $\Qbf^v_{\km1}$ has the components,
\[
    \left[ \Qbf^v_{\km1} \right]_{ij} = \Exp \left[ V_iV_j \right],
\]
where $V_i$ is the random variable in \eqref{eq:Velldef}.
Finally, the expectation for the cross-terms are given by
\begin{align*}
    \Exp[V_iX_j]
    &\stackrel{(a)}{=}
    \Exp[g_p(P_i,W^p,\gammabar_{1i},\alphabar_{1i})X_j] \nonumber \\
    &\stackrel{(b)}{=} \Exp\left[ g_p'(P_i,W^p,\gammabar_{1i},\alphabar_{1i})\right]
        \Exp[X_iX_j] \nonumber \\
     &\stackrel{(c)}{=}
     \Exp[X_iX_j]C_1(\alphabar_{1i})\left( \Exp
        \left[ f_p'(P_i,W^p,\gammabar_{1i}) \right]-\alphabar_{1i}
        \right) \\
     &\stackrel{(d)}{=} 0,
\end{align*}
where (a) follows from \eqref{eq:Velldef};
(b) follows from Stein's Lemma; and  (c) follows from the definition of $g_p(\cdot)$
in \eqref{eq:gpdef};  and (d) follows from \eqref{eq:a1segen}.
The above calculations show that
\beq \label{eq:AAlim}
    \lim_{N \arr \infty} \frac{1}{N}\Abf\tran_k\Abf_k 
        \stackrel{a.s.}{=} \left[ \begin{array}{cc}
        \Qbf_k^p & \mathbf{0} \\
        \mathbf{0} & \Qbf^v_{\km1}
        \end{array} \right],
\eeq
A similar calculation shows that
\beq \label{eq:Aslim}
    \lim_{N \arr \infty} \frac{1}{N} \Abf\tran_k\sbf_k = \left[
    \begin{array}{c} \mathbf{0} \\ \bbf^s_k \end{array} \right],
\eeq
where \textb{$\bbf^v_k$} is the vector of correlations
\beq
    \bbf^v_k = \mat{ \Exp[V_0V_k] & \Exp[V_1V_k] & \cdots & \Exp[V_{\km1}V_k] }\tran.
\eeq
Combining \eqref{eq:AAlim} and \eqref{eq:Aslim} shows that
\beq \label{eq:Vsmult1}
    \lim_{N \arr \infty} (\Abf\tran_k\Abf_k)^{-1}\Abf_k\tran \vbf_k 
    \stackrel{a.s.}{=}
    \left[ \begin{array}{c} \mathbf{0} \\ \mathbf{\beta}_k \end{array} \right],
\eeq
where 
\[
    \mathbf{\beta}_k := \left[ \Qbf^v_{k-1}\right]^{-1}\bbf^v_k.
\]
Therefore,
\begin{align}
    \MoveEqLeft \qbf_k^{\rm det} = \Bbf_k(\Abf\tran_k\Abf_k)^{-1}\Abf_k\tran\vbf_k \nonumber \\
    &= \left[ \Ubf_k ~ \Qbf_{\km1} \right]
    \left[ \begin{array}{c} \mathbf{0} \\ \mathbf{\beta}_k \end{array} \right]
    +\xibf \nonumber \\
    &= \sum_{\ell=0}^{\km1} \beta_{k\ell} \qbf_\ell + \xibf, \label{eq:qbetasum}
\end{align}
where \textb{$\xibf \in \R^N$ is the error,
\beq \label{eq:xisdef}
    \xibf = \Bbf_k \sbf, \quad
    \sbf := (\Abf\tran_k\Abf_k)^{-1}\Abf_k\tran \vbf_k -
    \left[ \begin{array}{c} \mathbf{0} \\ \mathbf{\beta}_k \end{array} \right].
\eeq
We next need to bound the norm of the error term $\xibf$.
Since $\xibf = \Bbf_k\sbf$, the definition of $\Bbf_k$ in 
\eqref{eq:ABdef} shows that
\beq \label{eq:xisum}
    \xibf = \sum_{i=0}^k s_i \ubf_i + \sum_{j=0}^{\km1} s_{k+j+1}\qbf_j,
\eeq
where we have indexed the components of $\sbf$ in \eqref{eq:xisdef}
as $\sbf=(s_0,\ldots,s_{2k})$.
From \eqref{eq:Vsmult1}, the components $s_j \arr 0$ almost surely,
and therefore
\[
    \lim_{N \arr \infty} \max_{j=0,\ldots,2k} |s_j| \stackrel{a.s.}{=} 0.
\]
Also, by the induction hypothesis,
\[
    \lim_{N \arr \infty} \frac{1}{N} \|\ubf_i\|^2 \stackrel{a.s.}{=} E(U_i^2),
    \quad
    \lim_{N \arr \infty} \frac{1}{N} \|\qbf_j\|^2 \stackrel{a.s.}{=} E(Q_j^2).
\]
Therefore, from \eqref{eq:xisum},
\begin{align}
    \MoveEqLeft \lim_{N \arr \infty} \frac{1}{N}\|\xibf\|^2
       \leq
        \lim_{N \arr \infty} 
        \left[ \max_{j=0,\ldots,2k} |s_j|^2 \right] \nonumber \\
        & \times 
        \frac{1}{N} \left[ \sum_i \|\ubf_i\|^2 + \sum_j \|\qbf_j\|^2 \right] \stackrel{a.s.}{=} 0.
    \label{eq:xilimbnd}
\end{align}
Therefore, if $f(q_1,\cdots,q_k)$ is
pseudo-Lipschitz continuous of order 2,
\begin{align*}
    \MoveEqLeft \lim_{N \arr \infty} \frac{1}{N} \sum_{n=1}^N
        f(q_{0n},\cdots,q_{\km1,n},q^{\rm det}_k) \\
        &\stackrel{(a)}{=} \lim_{N \arr \infty} \frac{1}{N} \sum_{n=1}^N
        f\left( q_{0n},\cdots,q_{\km1,n},
            \sum_{\ell=0}^{\km1} \beta_{k\ell} q_{\ell n} \right)\\
        &\stackrel{(b)}{=} \Exp\left[
        f\left(Q_0,\cdots,Q_{\km1}, \sum_{\ell=0}^{\km1} \beta_{k\ell} Q_{\ell n}
            \right) \right],
\end{align*}
where (a) follows from the \eqref{eq:qbetasum}, the bound
\eqref{eq:xilimbnd}, and the pseudo-Lipschitz continuity of $f(\cdot)$;
and (b) follows from the fact that $f(\cdot)$ is pseudo-Lipschitz continuous
and the induction hypothesis that
\[
    \lim_{N \arr \infty} \{q_{0n},\cdots,q_{\km1,n}\} \stackrel{PL(2)}{=}
    (Q_0,\ldots,Q_{\km1}).
\]
This proves \eqref{eq:qconvdet}.
}

\end{proof}

\begin{lemma} \label{lem:rhoconv}
Under the induction hypothesis $H_{k,\km1}$, the following limit
holds almost surely
\beq
    \lim_{N \arr \infty} \frac{1}{N} \| \Ubf_{\Abf_k^\perp}\tran\sbf_k\|^2 =
    \rho_k,
\eeq
for some constant $\rho_k \geq 0$.
\end{lemma}
\begin{proof}  From \eqref{eq:ABdef}, the matrix  $\Abf_k$ has $s=2k+1$ columns.
From Lemma~\ref{lem:orthogLin},
$\Ubf_{\Abf_k^\perp}$ is an orthonormal basis of $N-s$ in the $\mathrm{Range}(\Abf_k)^\perp$.
Hence, the energy $\| \Ubf_{\Abf_k^\perp}\sbf_k\|^2$ is precisely
\[
    \| \Ubf_{\Abf_k^\perp}\sbf_k\|^2 = \sbf_k\tran\sbf_k - \sbf_k\tran\Abf_k
        (\Abf_k\tran\Abf_k)^{-1}\Abf_k\tran\sbf_k.
\]
Using similar calculations as the previous lemma, we have
\[
    \lim_{N \arr \infty} \frac{1}{N} \| \Ubf_{\Abf_k}\sbf_k\|^2
    = \Exp[S_k^2] - (\bbf^s_k)\tran\left[ \Qbf^s_k \right]^{-1}\bbf^s_k.
\]
Hence, the lemma is proven if we define $\rho_k$ as the right hand side of this
equation.
\end{proof}

\begin{lemma} \label{lem:qconvran}
Under the induction hypothesis $H_{k,\km1}$, the components of
the ``random" part $\qbf_k^{\rm ran}$ along with the components
of $(\wbf^q,\qbf_0,\ldots,\qbf_{\km1})$
almost surely converge empirically as
\begin{align}
    \MoveEqLeft  \lim_{N \arr \infty}
        \left\{ (w^q_n,q_{0n},\ldots,q_{\km1,n},q_{kn}^{\rm ran}) \right\}
        \nonumber \\
        &\PLeq  (W^q,Q_0,\ldots,Q_{\km1},U_k), \label{eq:qranlim}
\end{align}
where $U_k \sim \Norm(0,\rho_k)$ is a Gaussian random variable
independent of $(W^q,Q_0,\ldots,Q_{\km1})$ and $\rho_k$ is the constant
in Lemma~\ref{lem:rhoconv}.
\end{lemma}
\begin{proof}
This is a direct application of Lemma~\ref{lem:orthogGaussLim}.
Let $\xbf = \tilde{\Vbf}\tran\Ubf_{\Abf_k^\perp}\tran\sbf_k$ so that
\[
    \qbf_k^{\rm det} = \Ubf_{\Bbf_k^\perp}\xbf_k.
\]
For each $N$, $\Ubf_{\Bbf_k^\perp} \in \R^{N \x (N-s)}$ is a matrix
with orthonormal columns spanning $\Range(\Bbf_k)^\perp$.
Also, since $\tilde{\Vbf}$ is uniformly distributed on the set of
$(N-s)\x (N-s)$ orthogonal matrices, and independent of $G_k$,
the conditional distribution $\xbf_k$ given $G_k$ is orthogonally invariant in that
\[
    \left. \Ubf \xbf_k \right|_{G_k} \eqd \left. \xbf_k \right|_{G_k},
\]
for any orthogonal matrix $\Ubf$.  Lemma~\ref{lem:rhoconv} also shows that
\[
    \lim_{N \arr \infty} \frac{1}{N} \|\xbf_k\|^2 = \rho_k,
\]
almost surely.
The limit \eqref{eq:qranlim} now follows from
Lemma~\ref{lem:orthogGaussLim}.
\end{proof}

Using the partition \eqref{eq:qpart} and Lemmas~\ref{lem:qconvdet} and \ref{lem:qconvran},
the components of $(\wbf^q,\qbf_0,\ldots,\qbf_k)$
almost surely converge empirically as
\begin{align*}
    \lefteqn{ \lim_{N \arr \infty} \{ (w^q_n,q_{0n},\ldots,q_{kn}) \} }\\
    &\PLeq \lim_{N \arr \infty} \{ (w^q_n,q_{0n},\ldots,q^{\rm det}_{kn} + q^{\rm ran}_{kn}) \}
    \\
    &\PLeq  (W^q,Q_0,\ldots,Q_k),
\end{align*}
where $Q_k$ is the random variable
\[
    Q_k = \beta_{k0}Q_0 + \cdots + \beta_{k,\km1}Q_{\km1} + U_k.
\]
Since $(Q_0,\ldots,Q_{\km1})$ is jointly Gaussian and $U_k$ is Gaussian independent of
$(Q_0,\ldots,Q_{\km1})$ we have that $(Q_0,\ldots,Q_k)$ is Gaussian.  This proves
\eqref{eq:Qconk}.

Now the function $\Gamma_1(\gamma_1,\alpha_1)$ is assumed to be
continuous at $(\gammabar_{1k},\alphabar_{1k})$.  Also, the induction hypothesis assumes that
$\alpha_{1k} \arr \alphabar_{1k}$ and $\gamma_{1k} \arr \gammabar_{1k}$ almost surely.
Hence,
\beq \label{eq:gam2limpf}
    \lim_{N \arr \infty} \gamma_{2k} = \lim_{N \arr \infty} \Gamma_1(\gamma_{1k},\alpha_{1k})
    = \gammabar_{2k}.
\eeq
In addition, since we have assumed that $\fbf_q'(\qbf,\wbf^q,\gamma_1)$ is Lipschitz
continuous in $(\qbf,\wbf^q)$ and continuous in $\gamma_1$,
\begin{align}
    \lim_{N \arr \infty} \alpha_{2k}
    &= \lim_{N \arr \infty} \bkt{\fbf_q'(\qbf_k,\wbf^q,\gamma_{1k})} \nonumber \\
    &= \Exp\left[ f_q'(Q_k,W^q,\gammabar_{1k}) \right] = \alphabar_{1k}. \label{eq:a2limpf}
\end{align}
The limits \eqref{eq:gam2limpf} and \eqref{eq:a2limpf} prove \eqref{eq:ag2limgen}.

Finally, we need to show that $\Exp[Q_k^2] = \tau_{2k}$ is the variance from the SE
equations.
\begin{align}
 \Exp[Q_k^2] &\stackrel{(a)}{=} \lim_{N \arr \infty} \frac{1}{N}
        \|\qbf_k\|^2 \nonumber \\
         & \stackrel{(b)}{=} \lim_{N \arr \infty} \frac{1}{N}
        \|\vbf_k\|^2 \nonumber \\
         & \stackrel{(c)}{=} \Exp\left[ g_p(P_k,W^p,\gammabar_{1k},\alphabar_{1k}) \right] \nonumber \\
         & \stackrel{(d)}{=} C_1^2(\alphabar_{1k})
         \Exp\left[ \left(f_p(P_k,W^p,\gammabar_{1k}) - \alphabar_{1k}P_k\right)^2 \right] \nonumber \\
         & = C_1^2(\alphabar_{1k})\Bigl\{
         \Exp\left[ f_p^2(P_k,W^p,\gammabar_{1k})\right]  \nonumber \\
         & \quad - 2\alphabar_{1k}
            \Exp\left[ P_k f_p(P_k,W^p,\gammabar_{1k})\right] + \alphabar^2_{1k}\Exp\left[ P_k^2 \right]
            \Bigr\} \nonumber \\
         & \stackrel{(e)}{=} C_1^2(\alphabar_{1k})\Bigl\{
         \Exp\left[ f_p^2(P_k,W^p,\gammabar_{1k})\right]  \nonumber \\
         & \quad - 2\alphabar_{1k}\tau_{1k}
            \Exp\left[ f_p'(P_k,W^p,\gammabar_{1k})\right] + \alphabar^2_{1k}\tau_{1k}   \Bigr\}
            \nonumber \\
         & \stackrel{(f)}{=} C_1^2(\alphabar_{1k})\left\{
         \Exp\left[ f_p^2(P_k,W^p,\gammabar_{1k})\right]- \alphabar_{1k}^2\tau_{1k} \right\} \nonumber \\
         & \stackrel{(g)}{=} \tau_{2k},
\end{align}
where (a) follows from the fact that the components of $\qbf_k$ converge empirically
to $Q_k$;
(b) follows from \eqref{eq:qupgen} and the fact that $\Vbf$ is orthogonal;
(c) follows from the limit \eqref{eq:Velldef}; and
(d) follows from \eqref{eq:gpdef};
(e) follows from Stein's Lemma and the fact that $\Exp [P_k^2] = \tau_{1k}$;
(f) follows from the definition of $\alphabar_{1k}$ in \eqref{eq:a1segen};
and (g) follows from \eqref{eq:tau2segen}.  Thus, $\Exp [Q_k^2] = \tau_{2k}$,
and we have proven the implication $H_{k,\km1} \Rightarrow H_{k,k}$.

\section{Proof of Theorem~\ref{thm:se} }

Theorem~\ref{thm:se} is essentially a special case of Theorem~\ref{thm:genConv}.
We need to simply rewrite the recursions in Algorithm~\ref{algo:vamp} in the form
\eqref{eq:algoGen}.
To this end, define the error terms
\beq \label{eq:pvslr}
    \pbf_k := \rbf_{1k}-\xbf^0, \quad
    \vbf_k := \rbf_{2k}-\xbf^0,
\eeq
and their transforms,
\beq \label{eq:uqslr}
    \ubf_k := \Vbf\tran\pbf_k, \quad
    \qbf_k := \Vbf\tran\vbf_k.
\eeq
Also, define the disturbance terms
\beq \label{eq:wpqslr}
    \wbf^q := (\xibf,\sbf), \quad
    \wbf^p := \xbf^0, \quad \xibf := \Ubf\tran\wbf,
\eeq
and the componentwise update functions
\begin{subequations} \label{eq:fqpslr}
\begin{align}
    f_q(q,(\xi,s),\gamma_2) &:= \frac{\gamma_w s\xi + \gamma_2 q}{
        \gamma_w s^2 + \gamma_2}, \label{eq:fqslr} \\
    f_p(p,x^0,\gamma_1) &= g_1(p+x^0,\gamma_1) - x^0.
        \label{eq:fpslr}
\end{align}
\end{subequations}
With these definitions, we claim that the outputs satisfy the recursions:
\begin{subequations} \label{eq:gecslr}
\begin{align}
    \pbf_k &= \Vbf\ubf_k \label{eq:pupslr} \\
    \alpha_{1k} &= \bkt{ \fbf_p'(\pbf_k,\xbf^0,\gamma_{1k})},
    \quad \gamma_{2k} = \frac{(1-\alpha_{1k})\gamma_{1k}}{\alpha_{1k}}
        \label{eq:alpha1slr}  \\
    \vbf_k &= \frac{1}{1-\alpha_{1k}}\left[
        \fbf_p(\pbf_k,\xbf^0,\gamma_{1k})- \alpha_{1k} \pbf_{k} \right] \label{eq:vupslr} \\
    \qbf_k &= \Vbf\tran\vbf_k \label{eq:qupslr} \\
    \alpha_{2k} &= \bkt{ \fbf_q'(\qbf_k,\wbf^q,\gamma_{2k})},
    \quad \gamma_{1,\kp1} = \frac{(1-\alpha_{2k})\gamma_{2k}}{\alpha_{2k}}
        \label{eq:alpha2slr} \\
    \ubf_{\kp1} &= \frac{1}{1-\alpha_{2k}}\left[
        \fbf_q(\qbf_k,\wbf^q,\gamma_{2k}) - \alpha_{2k}\qbf_{k} \right] \label{eq:uupslr}
\end{align}
\end{subequations}
Before we prove \eqref{eq:gecslr}, we can see that \eqref{eq:gecslr} is a special
case of the general recursions in \eqref{eq:algoGen} if we define
\[
    C_i(\alpha_i) = \frac{1}{1-\alpha_i}, \quad \Gamma_i(\gamma_i,\alpha_i) =
        \gamma_i\left[\frac{1}{\alpha_i}-1 \right].
\]
It is also straightforward to verify the continuity assumptions in Theorem~\ref{thm:genConv}.
The assumption of Theorem~\ref{thm:se} states that $\alphabar_{ik} \in (0,1)$.  Since
$\gammabar_{10} > 0$, $\gammabar_{ik} > 0$ for all $k$ and $i$.  Therefore,
$C_i(\alpha_i)$ and $\Gamma_i(\gamma_i,\alpha_i)$ are continuous at all points
$(\gamma_i,\alpha_i) = (\gammabar_{ik},\alphabar_{ik})$.
Also, since $s \in [0,S_{max}]$ and $\gamma_{2k} > 0$ for all $k$, the function
$f_q(q,(\xi,s),\gamma_2)$ in \eqref{eq:fqpslr} is uniformly Lipschitz continuous
in $(q,\xi,s)$ at all $\gamma_2 = \gammabar_{2k}$.
Similarly, since the denoiser function $g_1(r_1,\gamma_1)$ is assumed be to uniformly
Lipschitz continuous in $r_1$ at all $\gamma_1 = \gammabar_{1k}$, so is the function
$f_p(r_1,x^0,\gamma_1)$ in \eqref{eq:fpslr}.  Hence all the conditions of Theorem~\ref{thm:genConv}
are satisfied.  The SE equations \eqref{eq:se} immediately
from the general SE equations \eqref{eq:segen}.  In addition, the limits
\eqref{eq:limrx1} and and \eqref{eq:limqxi} are special cases of the limits
\eqref{eq:Pconk} and \eqref{eq:Qconk}.  This proves Theorem~\ref{thm:se}.

So, it remains only to show that the updates in
\eqref{eq:gecslr} indeed hold.
Equations \eqref{eq:pupslr} and \eqref{eq:qupslr} follow immediately
from the definitions \eqref{eq:pvslr} and \eqref{eq:uqslr}.
Next, observe that we can rewrite the
LMMSE estimation function \eqref{eq:g2slr} as
\begin{align}
    \lefteqn{ \gbf_2(\rbf_{2k},\gamma_{2k}) }\nonumber\\
    &\stackrel{(a)}{=} \left( \gamma_w \Abf\tran\Abf + \gamma_{2k}\Ibf\right)^{-1}
    \left( \gamma_w\Abf\tran\Abf\xbf^0 + \gamma_w\Abf\tran\wbf
    + \gamma_{2k}\rbf_{2k} \right) \nonumber \\
   &\stackrel{(b)}{=} \xbf^0 +
   \left( \gamma_w \Abf\tran\Abf + \gamma_{2k}\Ibf\right)^{-1}
    \left( \gamma_{2k}(\rbf_{2k}-\xbf^0)+ \gamma_w\Abf\tran\wbf\right) \nonumber \\
    &\stackrel{(d)}{=} \xbf^0 +
    \Vbf\left( \gamma_w \Sbf^2 + \gamma_{2k}\Ibf\right)^{-1}
    \left( \gamma_{2k}\qbf_k + \Sbf\xibf \right), \nonumber  \\
    &\stackrel{(d)}{=} \xbf^0 + \Vbf \fbf_q(\qbf_k,\wbf^q,\gamma_{2k}),   \label{eq:g2slrV}
\end{align}
where (a) follows by substituting \eqref{eq:yAxslr} into \eqref{eq:g2slr};
(b) is a simple algebraic manipulation;
(c) follows from the SVD definition \eqref{eq:ASVD} and the definitions
$\xibf$ in \eqref{eq:wpqslr} and $\qbf_k$ in \eqref{eq:uqslr}; and
(d) follows from the definition of componentwise function
$f_q(\cdot)$ in \eqref{eq:fqslr}. Therefore, the divergence $\alpha_{2k}$ satisfies
\begin{align}
    \alpha_{2k}
    &\stackrel{(a)}{=}
    \frac{1}{N}
        \Tr\left[ \frac{\partial \gbf_2(\rbf_{2k},\gamma_{2k})}{\partial \rbf_{2k}} \right]
        \nonumber \\
    & \stackrel{(b)}{=}    \frac{1}{N}
        \Tr\left[ \Vbf \Diag(\fbf_q'(\qbf_{k},\wbf^q,\gamma_{2k}))
        \frac{\partial \qbf_k}{\partial \rbf_{2k}}  \right]
        \nonumber \\
    & \stackrel{(c)}{=}    \frac{1}{N}
        \Tr\left[ \Vbf \Diag(\fbf_q'(\qbf_{k},\wbf^q,\gamma_{2k})) \Vbf\tran  \right] \nonumber \\
    & \stackrel{(d)}{=}  \bkt{ \fbf_q'(\qbf_{k},\wbf^q,\gamma_{2k}) },
    \label{eq:a2pf}
\end{align}
where
(a) follows from line~\ref{line:a2} of Algorithm~\ref{algo:vamp} and \eqref{eq:jacobian}--\eqref{eq:bkt};
(b) follows from \eqref{eq:g2slrV}; (c) follows from \eqref{eq:uqslr}; and
(d) follows from $\Vbf\tran\Vbf=\Ibf$ and \eqref{eq:jacobian}--\eqref{eq:bkt}.
Also, from lines~\ref{line:eta2}-\ref{line:gam1} of Algorithm~\ref{algo:vamp},
\beq \label{eq:g2upa}
    \gamma_{1,\kp1} = \eta_{2k}-\gamma_{2k} = \gamma_{2k}\left[ \frac{1}{\alpha_{2k}}-1 \right].
\eeq
Equations \eqref{eq:a2pf} and \eqref{eq:g2upa} prove \eqref{eq:alpha2slr}.
In addition,
\begin{align}
    \lefteqn{ \pbf_{\kp1} \stackrel{(a)}{=} \rbf_{1,\kp1} - \xbf^0 }\nonumber\\
    &\stackrel{(b)}{=}  \frac{1}{1-\alpha_{2k}}\left[
        \gbf_2(\rbf_{2k},\gamma_{2k}) - \alpha_{2k}\rbf_{2k} \right]
        -\xbf^0 \nonumber \\
    &\stackrel{(c)}{=}  \frac{1}{1-\alpha_{2k}}\left[
        \xbf^0 + \Vbf \fbf_q(\qbf_k,\wbf^q,\gamma_{2k})
        - \alpha_{2k}(\xbf^0 + \vbf_{k}) \right]
        -\xbf^0 \nonumber \\
    &\stackrel{(d)}{=}  \frac{1}{1-\alpha_{2k}}\left[
        \Vbf \fbf_q(\qbf_k,\wbf^q,\gamma_{2k})
        - \alpha_{2k}\vbf_{k} \right] \nonumber \\
    &\stackrel{(e)}{=}  \Vbf \left[ \frac{1}{1-\alpha_{2k}}\left[
        \fbf_q(\qbf_k,\wbf^q,\gamma_{2k}) -  \alpha_{2k}\qbf_{k} \right] \right]
        \label{eq:pupslr2} ,
\end{align}
where (a) follows from \eqref{eq:pvslr};
(b) follows from 
lines~\ref{line:x2}-\ref{line:r1} of Algorithm~\ref{algo:vamp};
(c) follows from \eqref{eq:g2slrV} and the definition of $\vbf_k$ in \eqref{eq:pvslr};
(d) follows from collecting the terms with $\xbf^0$;
and (e) follows from the definition
$\qbf_k=\Vbf\tran\vbf_k$ in \eqref{eq:uqslr}.
Combining \eqref{eq:pupslr2} with
$\ubf_{\kp1} = \Vbf\tran\pbf_{\kp1}$ proves \eqref{eq:uupslr}.

The derivation for the updates for $\vbf_k$ are similar.  First,
\begin{align}
    \MoveEqLeft \alpha_{1k} \stackrel{(a)}{=}
    \bkt{ \gbf_1'(\rbf_{1k},\gamma_{1k}) }
     \stackrel{(b)}{=}  \bkt{ \fbf_p'(\pbf_{k},\xbf^0) },
    \label{eq:a1pf}
\end{align}
where (a) follows from line~\ref{line:a1} of Algorithm~\ref{algo:vamp}
and (b) follows from the vectorization
of $\fbf_p(\cdot)$ in \eqref{eq:fpslr} and the fact that $\pbf_k=\rbf_{1k}+\xbf^0$.
Also, from lines~\ref{line:eta1}-\ref{line:gam2} of Algorithm~\ref{algo:vamp},
\beq \label{eq:g1upa}
    \gamma_{2k} = \eta_{1k}-\gamma_{1k} = \gamma_{1k}\left[ \frac{1}{\alpha_{1k}}-1 \right].
\eeq
Equations \eqref{eq:a1pf} and \eqref{eq:g1upa} prove \eqref{eq:alpha1slr}.
Also,
\begin{align}
    \lefteqn{ \vbf_{k} \stackrel{(a)}{=} \rbf_{2k} - \xbf^0 } \nonumber \\
    &\stackrel{(b)}{=}  \frac{1}{1-\alpha_{1k}}\left[
        \gbf_1(\rbf_{1k},\gamma_{1k}) - \alpha_{1k}\rbf_{1k} \right]   -\xbf^0 \nonumber \\
    &\stackrel{(c)}{=}  \frac{1}{1-\alpha_{1k}}\left[
        \fbf_p(\pbf_{k},\xbf^0,\gamma_{1k}) +\xbf^0 - \alpha_{1k}(\pbf_{k}+\xbf^0) \right]
     -\xbf^0 \nonumber \\
    &\stackrel{(d)}{=}  \frac{1}{1-\alpha_{1k}}\left[
        \fbf_p(\pbf_{k},\xbf^0,\gamma_{1k}) - \alpha_{1k}\pbf_{k}\right]
\end{align}
where (a) is the definition of $\vbf_k$ in \eqref{eq:pvslr};
(b) follows from lines~\ref{line:x1}-\ref{line:r2} of Algorithm~\ref{algo:vamp};
(c) follows from the vectorization of $f_p(\cdot)$ in \eqref{eq:fpslr} and the definition of $\pbf_k$ in \eqref{eq:pvslr};
and (d) follows from collecting the terms with $\xbf^0$.
This proves \eqref{eq:vupslr}.  All together, we have proven \eqref{eq:gecslr} and
the proof is complete.

\section{Proof of Theorem~\ref{thm:seMmse}} \label{sec:seMmsePf}
We use induction.  Suppose that, for some $k$,
$\gammabar_{1k} = \tau_{1k}^{-1}$.  From \eqref{eq:a1se}, \eqref{eq:A1match}
and \eqref{eq:E1match},
\beq \label{eq:a1matchpf}
    \alphabar_{1k} = \gammabar_{1k}\Ecal_1(\gammabar_{1k}).
\eeq
Hence, from \eqref{eq:eta1se}, $\etabar_{1k}^{-1} = \Ecal_1(\gammabar_{1k})$ and
$\gammabar_{2k} = \etabar_{1k} - \gammabar_{1k}$.  Also,
\begin{align*}
      \tau_{2k}
      &\stackrel{(a)}{=} \frac{1}{(1-\alphabar_{1k})^2}\left[
        \Ecal_1(\gammabar_{1k},\tau_{1k}) - \alphabar_{1k}^2\tau_{1k} \right] \\
      &\stackrel{(b)}{=} \frac{1}{(1-\gammabar_{1k}\Ecal_1(\gammabar_{1k}))^2}\left[
        \Ecal_1(\gammabar_{1k},\tau_{1k}) - \gammabar_{1k}\Ecal_1^2(\gammabar_{1k}) \right] \\
      &\stackrel{(c)}{=} \frac{\Ecal_1(\gammabar_{1k},\tau_{1k})}
      {1-\gammabar_{1k}\Ecal_1(\gammabar_{1k})} \\
      &\stackrel{(d)}{=} \frac{1}{\etabar_{1k} - \gammabar_{1k}},
\end{align*}
where (a) follows from \eqref{eq:tau2se};
(b) follows from \eqref{eq:a1matchpf} and the matched condition $\gammabar_{1k} = \tau_{1k}^{-1}$;
(c) follows from canceling terms in the fraction and (d) follows from the fact that
$\etabar_{1k}^{-1} = \Ecal_1(\gammabar_{1k})$ and $\gammabar_{1k} = \etabar_{1k}/\alphabar_{1k}$.
This proves \eqref{eq:eta1sematch}.  A similar argument shows that \eqref{eq:eta2sematch} holds if
$\gammabar_{2k} = \tau_{2k}^{-1}$.    Finally, \eqref{eq:etammse} follows from
\eqref{eq:sematch} and \eqref{eq:mseEcal}.

\bibliographystyle{IEEEtran}
\bibliography{../bibl}
\end{document}